\newcommand{\longversion}[1]{#1}
\newcommand{\shortversion}[1]{}
\newcommand{\trash}[1]{}

\longversion{
\documentclass[a4paper,USenglish]{article}
\usepackage[hscale=0.7,scale=0.75]{geometry}
}
\shortversion{
\documentclass[a4paper,USenglish]{lipics-v2016}}

\longversion{%
  \usepackage[USenglish]{babel}%
}%

\newcommand{\footnoteitext}[1]{\stepcounter{footnote}
  \footnotetext[\thefootnote]{#1}}

\usepackage{paralist}
\usepackage[inline]{enumitem}

\def\etal{et~al.{}}

\usepackage{amssymb,amsfonts,amsmath}%

\newcommand{\Ra}{\ensuremath{\Rightarrow}}

\newcommand{\La}{\ensuremath{\Leftarrow}}
\newcommand{\pif}{\noindent(\La)}
\newcommand{\ponlyif}{\noindent(\Ra)}
\newcommand{\rsep}{\nobreak\ensuremath{;\;}}
\newcommand{\pnot}{\nobreak\ensuremath{\neg}}

\DeclareFontFamily{U}{matha}{\hyphenchar\font45}
\DeclareFontShape{U}{matha}{m}{n}{
      <5> <6> <7> <8> <9> <10> gen * matha
      <10.95> matha10 <12> <14.4> <17.28> <20.74> <24.88> matha12
      }{}
\DeclareSymbolFont{matha}{U}{matha}{m}{n}
\DeclareMathSymbol{\squplus}{2}{matha}{"5D}

\makeatletter
\newcommand{\raisemath}[1]{\mathpalette{\raisem@th{#1}}}
\newcommand{\raisem@th}[3]{\raisebox{#1}{$#2#3$}}
\makeatother

\usepackage{stmaryrd}%

\makeatletter
\newcommand{\pushright}[1]{\ifmeasuring@#1\else\omit\hfill\ensuremath{\displaystyle#1}\fi\ignorespaces}
\newcommand{\pushleft}[1]{\ifmeasuring@#1\else\omit$\displaystyle#1$\hfill\fi\ignorespaces}
\providecommand{\leftsquigarrow}{%
	\mathrel{\mathpalette\reflect@squig\relax}%
}
\newcommand{\reflect@squig}[2]{%
	\reflectbox{$\m@th#1\rightsquigarrow$}%
}

\makeatother
\DeclareMathOperator{\disj}{DISJ}%
\DeclareMathOperator{\choice}{CH}%
\DeclareMathOperator{\opt}{OPT}%
\DeclareMathOperator{\optimize}{\leftsquigarrow}

\DeclareMathOperator{\type}{type}
\newcommand{\intr}{\textit{int}}
\newcommand{\leaf}{\textit{leaf}}
\newcommand{\rem}{\textit{rem}}
\newcommand{\join}{\textit{join}}

\DeclareMathOperator{\SSR}{SatPr}

\DeclareMathOperator{\post}{post-order}

\newcommand{\AlgA}{\ensuremath{\mathbb{A}}}%
\newcommand{\AlgC}{\ensuremath{\mathbb{C}}}%
\newcommand{\AlgW}{\ensuremath{\mathbb{W}}}%

\DeclareMathOperator{\E}{E}

\newcommand{\BIGOP}[1]
{
\mathop{\mathchoice%
{\raise-0.22em\hbox{\huge $#1$}}%
{\raise-0.05em\hbox{\Large $#1$}}{\hbox{\large $#1$}}{#1}}}

\newcommand{\BIGboxplus}{\mathop{\mathchoice%
{\raise-0.35em\hbox{\huge $\boxplus$}}%
{\raise-0.15em\hbox{\Large $\boxplus$}}{\hbox{\large $\boxplus$}}{\boxplus}}}

\DeclareMathOperator{\compr}{compr}

\DeclareMathOperator{\orig}{orig}

\DeclareMathOperator{\subsets}{subs}
\DeclareMathOperator{\evol}{evol}
\DeclareMathOperator{\res}{res}

\DeclareMathOperator{\aug}{aug}
\DeclareMathOperator{\sol}{sol}

\newcommand{\Tab}[1]{\ensuremath{\text{Child-Tabs}}}
\newcommand{\Tabs}[1]{\ensuremath{{\AlgA}\text{-}\text{Tabs[$#1$]}}}
\newcommand{\ATabs}[2]{\ensuremath{#1\text{-}\text{Tabs[$#2$]}}}
\newcommand{\ATab}[1]{\ensuremath{#1\text{-}\text{Tabs}}}
\renewcommand{\Big}{} %

\newcommand{\CCC}{\ensuremath{\mathcal{C}}}%
\newcommand{\TTT}{\ensuremath{\mathcal{T}}}%

\usepackage{latexsym}
\usepackage{rotating}

\usepackage{paralist}

\usepackage{url}
\usepackage{xspace}

\newcommand{\ASP}{\textsc{ASP}\xspace}

\newcommand{\pname}[1]{\textsc{#1}\xspace}

\newcommand*\mcupinn[2]{\vcenter{\hbox{$\mathsurround=0pt
  \ifx\displaystyle#1\textstyle\else#1\fi\bigcup$}}}
\usepackage[]{xcolor}

\newcommand{\emptyfunc}{\emptyset}

\newcommand{\inputPredColor}{orange!55!red}
\newcommand{\outputPredColor}{blue!45!black}
\newcommand{\statePredColor}{green!62!black}
\newcommand{\specialPredColor}{red!62!black}

\usepackage{multirow}

\usepackage{graphicx}
\graphicspath{{fig/}}
\DeclareGraphicsExtensions{.pdf,.png,.jpg}
\usepackage{booktabs}

\usepackage{microtype}
\usepackage{algpseudocode}
\algrenewcommand\algorithmicensure{\textbf{Output:}}
\algrenewcommand{\algorithmiccomment}[1]{\emph{// #1}}

\usepackage[ruled,vlined,linesnumbered]{algorithm2e}
\SetKwInput{KwData}{In}
\SetKwInput{KwResult}{Out}
\SetAlFnt{\small}
\SetAlCapFnt{\small}
\SetAlCapNameFnt{\small}
\SetAlCapHSkip{0pt}
\SetEndCharOfAlgoLine{}
\IncMargin{-\parindent}

\usepackage{paralist}

\usepackage{float}
\usepackage[labelsep=quad,aboveskip=2pt]{caption}

\usepackage{subcaption}
\usepackage{csquotes}
\usepackage{paralist}
\usepackage{multicol}
\usepackage{multirow}
\usepackage{listings}

\newcommand{\tuplecolor}[1]{\textcolor{#1}}
\lstdefinelanguage{dflat}{
	numberstyle=\tiny,
	otherkeywords={:-},
	morekeywords={not},
	keywordstyle=\bfseries,
	emph={numChildNodes,initial,final,currentNode,childNode,bag,current,introduced,removed,atLevel,atNode,root,rootOf,leaf,leafOf,sub,childItem,childAuxItem,childCost,childOr,childAnd,childAccept,childReject,childRow},
	moreemph=[2]{item,auxItem,extend,cost,currentCost,length,or,and,accept,reject},
	alsoletter={\#}, %
	morecomment=[l]{\%},
	emphstyle=\color{\inputPredColor},
	emphstyle=[2]\color{\outputPredColor},
	literate={:-}{{$\leftarrow$}}2 {!=}{{$\neq$}}1, 
	breakindent=3em,
	escapechar=@, %
	captionpos=b,
        frame=bt, 
        numbers=left
}
 \usepackage{relsize}
\usepackage{tikz}
\usetikzlibrary{shapes}
\usetikzlibrary{calc}

\shortversion{%
  \theoremstyle{plain}
  \newtheorem{proposition}{Proposition}%
  \newtheorem{observation}{Observation}
}%
\longversion{%
   \usepackage{amsthm}
  \newtheorem{remark}{Remark}%
  \newtheorem{theorem}{Theorem}%
  \newtheorem{lemma}{Lemma}%
  \newtheorem{proposition}{Proposition}%
  \newtheorem{observation}{Observation}
  \newtheorem{example}{Example}%
  \newtheorem{definition}{Definition}%
  \newtheorem{corollary}{Corollary}%
}
\newcommand{\MAI}[1]{\ensuremath{#1^+_a}}%
\newcommand{\MAR}[1]{\ensuremath{#1^-_a}}%

\newcommand{\MARR}[2]{\ensuremath{#1^-_{#2}}}%
\newcommand{\MAIR}[2]{\ensuremath{#1^+_{#2}}}%
\newcommand{\tabval}{\ensuremath{u}}
\newcommand{\tab}[1]{\ensuremath{\tau_{#1}}}
\newcommand{\at}{\text{\normalfont at}}
\newcommand{\oo}{\text{\normalfont cst}}
\newcommand{\att}[1]{\ensuremath{\at_{\hspace{-0.05em}\leq\hspace{-0.05em}#1}}}
\newcommand{\atto}{\ensuremath{\att{t}}}
\newcommand{\progt}[1]{\ensuremath{\prog_{\hspace{-0.05em}\leq\hspace{-0.05em}#1}}}
\newcommand{\progtneq}[1]{\ensuremath{\prog_{\hspace{-0.05em}<\hspace{-0.05em}#1}}}

\newcommand{\dpa}{\ensuremath{\mathtt{DP}}}
\newcommand{\cw}{\ensuremath{\mathtt{DPL}}}
\newcommand{\por}{\vee}

\newcommand{\eqdef}{\ensuremath{\,\mathrel{\mathop:}=}}
\newcommand{\hsep}{\leftarrow\,}
\newcommand{\Card}[1]{|#1|}

\newcommand{\AspComp}{\pname{AS}} %
\newcommand{\AspCount}{\pname{\#Asp}} %
\newcommand{\AspEnum}{\pname{EnumAsp}} %

\newcommand{\MSat}{\pname{MinSAT1}}
\newcommand{\MSatEnum}{\pname{EnumMinSAT1}}

\newcommand{\AspCountO}{\pname{\AspCount{}O}} %

\newcommand{\PRIMSAT}{\ensuremath{{\algo{MOD}}}\xspace}

\newcommand{\INCSAT}{\ensuremath{{\algo{MOD}}}\xspace}
\newcommand{\INCCOUNTERSAT}{\ensuremath{{\algo{CMOD}}}\xspace}

\newcommand{\INC}{\ensuremath{{\algo{SINC}}}\xspace}
\newcommand{\IINC}{\ensuremath{{\algo{INC}}}\xspace}
\newcommand{\problemFont}[1]{\textsc{#1}}

\newlength\problemlength
\newcommand\dproblem[3]{%
\begin{center}
\fbox{%
\begin{minipage}{.93	\linewidth}%
\begin{list}{}{\labelwidth\problemlength \labelsep.7em \rightmargin1.5em
\leftmargin\problemlength \advance\leftmargin by3em%
\parsep0ex \itemsep.2ex plus.1ex}
\item[{\sl Problem:\hfill}] {\problemFont{#1}}
\item[{\sl Input:  \hfill}] #2
\item[{\sl Task: \hfill}] #3
\end{list}
\end{minipage}
}
\end{center}
}

\usepackage{url}
\usepackage{doi}

\newenvironment{indented}{\begin{changemargin}{1cm}{0cm}}{\end{changemargin}}

\let\phi\varphi
\let\epsilon\varepsilon

\renewcommand{\models}{\vDash}

\newcommand{\CCard}[1]{\|#1\|}

\newcommand{\algo}[1]{\ensuremath{\mathbb{#1}}}

\newcommand{\NP}{\ensuremath{\textsc{NP}}\xspace}

\newcommand{\SIGMA}[2]{\ensuremath{\Sigma_{\textrm{#1}}^{\textrm{#2}}}}

\newcommand{\bigO}[1]{\ensuremath{{\mathcal O}(#1)}}

\newcommand{\tw}[1]{\mathit{tw}(#1)}

\newcommand{\SB}{\{}%
\newcommand{\SM}{\mid}%
\newcommand{\SE}{\}}%
\def\hy{\hbox{-}\nobreak\hskip0pt}

\newcommand{\solver}[1]{\mbox{\text{#1}}\xspace}

\newcommand{\dynasp}[1]{\ensuremath{\solver{DynASP2}}}
\newcommand{\dynaspplus}[1]{\ensuremath{\solver{DynASP2.5}}}
\newcommand{\prog}{\ensuremath{P}}

\usetikzlibrary{positioning}
\usetikzlibrary{shapes,arrows}
\tikzstyle{arg}=[draw, thick, circle]

\colorlet{afnodecolor}{green!20!blue!10}
\colorlet{tdnodecolor}{green!20!blue!10}
\colorlet{subfwnodecolor}{black!2}
\colorlet{subfwafinactivenodecolor}{white}
\colorlet{vertexTopColor}{white}
\colorlet{vertexBottomColor}{black!10}

\tikzstyle{afnode} = [draw,thick,shape=circle,minimum size=8mm,font=\normalsize,fill=afnodecolor]
\tikzstyle{afedge} = [->,draw,thick]
\tikzstyle{tdnode} = [draw,rounded corners,top color=vertexTopColor,bottom color=vertexBottomColor,minimum size=1.5em]
\tikzstyle{stdnode} = [tdnode, font=\scriptsize]
\tikzstyle{stdnodecompact} = [stdnode, inner sep = 1.5pt, outer sep = 0.1pt]
\tikzstyle{stdnodetable} = [stdnode, inner sep = 1.5pt, outer sep = 0]
\tikzstyle{stdnodenum} = [minimum size=1.5em, font=\scriptsize]
\tikzstyle{tdedge} = [-,draw,thick]
\tikzstyle{tdlabel} = [draw=none, rectangle, fill=none, inner sep=0pt, font=\scriptsize]
\tikzstyle{subfwnode} = [draw,thick,shape=rectangle,thin,rounded corners,minimum size=9mm,fill=subfwnodecolor,label distance=-2.5mm]
\tikzstyle{subfwafactivenode} = [draw,thick,shape=circle,minimum size=6mm,inner sep = 0pt,font=\scriptsize,fill=afnodecolor]
\tikzstyle{subfwafinactivenode} = [draw,thick,shape=circle,minimum size=6mm,inner sep = 0pt,font=\scriptsize,fill=white,dotted]
\tikzstyle{subfwafinactiveedge} = [->,draw,thick,dotted]

\tikzstyle{itemTree}=[level distance=2em,sibling distance=4ex,child anchor=west,grow'=right,right,align=left,every node/.style={draw,dashed,draw opacity=0.2,font=\footnotesize}]
\tikzstyle{itemTreeRoot}=[solid,inner sep=2]
\tikzstyle{orNode}=[label=left:$\lor$]
\tikzstyle{andNode}=[label=left:$\land$]
\tikzstyle{acceptNode}=[label=right:$\top$]
\tikzstyle{rejectNode}=[label=right:$\bot$]

\newcommand{\MDP}{\ensuremath{\mathtt{M\hy DP}_{\INC}}\xspace}

\longversion{ \title{%
    DynASP2.5: Dynamic Programming on \\Tree Decompositions in
    Action%
    \thanks{This is the author’s self-archived copy including
      detailed proofs. Research was supported by
      the Austrian Science Fund (FWF), Grant Y698.}  }

\author{Johannes K. Fichte, Markus Hecher, Michael Morak, Stefan Woltran\\ 
TU Wien, Vienna, Austria\\
  \texttt{lastname@dbai.tuwien.ac.at}}}
  
\shortversion{\title{%
  DynASP2.5: Dynamic Programming on Tree Decompositions in
  Action%
  \footnote{Research was supported by
      the Austrian Science Fund (FWF), Grant Y698.}
}}%

\shortversion{
\author[1]{Johannes K. Fichte}
\author[2]{Markus Hecher}
\author[3]{Michael Morak}
\author[4]{Stefan Woltran}
\affil[1]{TU Wien, Vienna, Austria\\
  \texttt{lastname@dbai.tuwien.ac.at}}
\authorrunning{J.\,K. Fichte, M. Hecher, M. Morak, and S. Woltran}

\Copyright{Johannes\,K. Fichte, Markus Hecher, Michael Morak, and Stefan Woltran}

\subjclass{Dummy classification -- please refer to
  \url{http://www.acm.org/about/class/ccs98-html}}%
\keywords{ Parameterized algorithms; Fixed-parameter linear time;
  Semi\hy incidence graph; Tree decompositions, Multi-pass dynamic
  programming }%
\EventEditors{John Q. Open and Joan R. Acces}
\EventNoEds{2}
\EventLongTitle{The 12th International Symposium on Parameterized and Exact Computation (IPEC 2017)}
\EventShortTitle{IPEC 2017}
\EventAcronym{IPEC}
\EventYear{2017}
\EventDate{September..., 2017}
\EventLocation{Vienna, Austria}
\EventLogo{}
\SeriesVolume{}
\ArticleNo{}
}

\begin{document}

\label{firstpage}

\maketitle
\begin{abstract}%
  A vibrant theoretical research area are efficient exact
  parameterized algorithms. Very recent solving competitions such as
  the PACE challenge show that there is also increasing practical
  interest in the parameterized algorithms community. An important
  research question is whether dedicated parameterized exact
  algorithms exhibit certain practical relevance and one can even beat
  well-established problem solvers.
  We consider the logic-based declarative modeling language and
  problem solving framework Answer Set Programming (\ASP).
  State-of-the-art \ASP solvers rely considerably on
  \textsc{Sat}-based algorithms.  An \ASP solver (DynASP2), which is
  based on a classical dynamic programming on tree decompositions, has
  been published very recently.
  Unfortunately, DynASP2 can outperform modern \ASP solvers on
  programs of small treewidth only if the question of interest is to
  count the number of solutions.
  In this paper, we describe underlying concepts of our new
  implementation (DynASP2.5) that shows \emph{competitive behavior} to
  state-of-the-art \ASP solvers \emph{even} for \emph{finding just one
    solution} when solving problems as the Steiner tree problem that
  have been modeled in \ASP on graphs with low treewidth. Our
  implementation is based on a novel approach that we call multi-pass
  dynamic programming (\MDP).

\end{abstract}

\newcommand{\SAT}{\textsc{Sat}\xspace}
\newcommand{\UNSAT}{\textsc{Unsat}\xspace}
\section{Introduction}

Answer set programming (\emph{\ASP}) is a logic-based declarative
modelling language and problem solving
framework~\cite{JanhunenNiemela16a}, where a program consists of sets
of rules over propositional atoms and is interpreted under an extended
stable model semantics~\cite{SimonsNiemelaSoininen02}.
Problems are usually modelled in \ASP in such a way that the stable
models (answer sets) of a program directly form a solution to the
considered problem instance.
Computational problems for disjunctive, propositional \ASP such as
deciding whether a program has an answer set are complete for the
second level of the Polynomial Hierarchy~\cite{EiterGottlob95}.
In consequence, finding answer sets usually involves a \SAT part
(finding a model of the program) and an \UNSAT part (minimality check).
A variety of CDCL-based \ASP solvers have been
implemented~\shortversion{\cite{clasp}}\longversion{\cite{clasp,AlvianoEtAl13}} and proven to be very
successful in solving competitions\longversion{~\cite{GebserMarateaRicca16a}}.
Very recently, a dynamic programming based solver (\emph{DynASP2})
that builds upon ideas from parameterized algorithmics was
proposed~\cite{FichteEtAl17}.
For disjunctive input programs, the runtime of the underlying
algorithms is double exponential in the incidence treewidth and linear
in the input size (so-called \emph{fixed-parameter linear}
algorithms).
DynASP2~(i) takes a tree decomposition of a certain graph
representation (incidence graph) of a given input program and
(ii)~solves the program via dynamic programming (DP) on the tree
decomposition by traversing the tree exactly once.  Both finding a
model and checking minimality are considered at the same time. Once
the root node has been reached, complete solutions (if exist) for the
input program can be constructed.
This approach pays off for counting answer sets, but is \emph{not
  competitive for outputting just one answer set}. The reason for that
lies in the exhaustive nature of dynamic programming as all potential
values are computed \emph{locally} for each node of the tree
decomposition.
In consequence, space requirements can be quite extensive resulting in
long running times.
Moreover, dynamic programming algorithms on tree decompositions may
yield extremely diverging run-times on tree decompositions of the
exact same width~\cite{AbseherEtAl17a}.
In this paper, we
propose a multi-pass approach (\MDP) for dynamic programming on tree
decompositions as well as a new implementation
(DynASP2.5).  %
In contrast to classical dynamic programming algorithms for problems
on the second level of the Polynomial Hierarchy, \MDP traverses the
given tree decomposition multiple times. Starting from the leaves, we
compute and store (i)~sets of atoms that are relevant for the \SAT
part (finding a model of the program) up to the root. Then we go back
again to the leaves and compute and store (ii)~sets of atoms that are
relevant for the \UNSAT part (checking for minimality). Finally, we go
once again \emph{back to} the leaves and (iii)~link sets from past
Passes~(i) and (ii) that might lead to an answer set in \emph{the
  future}. As a result, we allow for early cleanup of candidates
that do not lead to answer sets.

Further, we present technical improvements (including working on
non-normalized tree decompositions) and employ dedicated
\emph{customization techniques} for selecting tree decompositions. Our
improvements are main ingredients to speedup the solving process for
DP algorithms.
Experiments indicate that DynASP2.5 is competitive even for
finding one answer set using the Steiner tree
problem on graphs with low treewidth. In particular, we are able
to solve instances that have an upper bound on the incidence treewidth
of~$14$ (whereas DynASP2 solved instances of treewidth at most~$9$).

\noindent
\emph{Our main contributions} can be summarized as follows:
\begin{enumerate}
\item We establish a novel fixed-parameter linear algorithm (\MDP),
  which works in multiple passes and computes \SAT and \UNSAT parts
  separately.
\item We present an implementation (DynASP2.5)\footnote{The source
    code of our solver is available at
    \url{https://github.com/daajoe/dynasp/releases/tag/v2.5.0}.} and
  an experimental evaluation. %
\end{enumerate}
\noindent \emph{Related Work.} Jakl, Pichler, and
Woltran~\cite{JaklPichlerWoltran09} have considered ASP solving when
parameterized by the treewidth of a graph representation and suggested
fixed-parameter linear algorithms.
Fichte~\etal~\cite{FichteEtAl17} have established additional
algorithms and presented empirical results on an implementation that
is dedicated to counting answer sets for the full ground ASP
language. The present paper extends their work by a multi-pass dynamic
programming algorithm.
Bliem~\etal~\cite{BliemEtAl16b} have introduced a general multi-pass
approach and an implementation (\text{D-FLAT}\^{}2) for dynamic
programming on tree decompositions solving subset minimization
tasks. Their approach allows to specify dynamic programming algorithms
by means of ASP. In a way, one can see ASP in their approach as a
meta-language to describe table algorithms\footnote{See
  Algorithm~\ref{fig:dpontd} for the concept of table algorithms.},
whereas our work presents a dedicated algorithm to find an answer set
of a program. In fact, our implementation extends their general ideas
for subset minimization (disjunctive rules) to also support weight
rules. However, due to space constraints we do not report on weight
rules in this paper. Beyond that, we require specialized adaptions to
the ASP problem semantics, including three valued evaluation of atoms,
handling of non-normalized tree decompositions, and optimizations in
join nodes to be competitive.
Abseher, Musliu, and Woltran~\cite{AbseherMusliuWoltran17a} have
presented a framework that computes tree decompositions via
heuristics, which is also used in our solver.  Other tree
decomposition systems can be found on the PACE challenge
website\longversion{~\cite{Dell17a}}. Note that improved heuristics for finding a
tree decomposition of smaller width (if possible) directly yields
faster results for our solver.

\section{Formal Background} %
\label{sec:preliminaries}%
\shortversion{\noindent \textit{Tree Decompositions.}}%
\longversion{\subsection{Tree Decompositions}}%
\longversion{%

} %
Let $G = (V,E)$ be a graph, $T = (N,F,n)$ a rooted tree, and
$\chi: N \to 2^V$ a function that maps each node~$t \in N$ to a set of
vertices. We call the sets $\chi(\cdot)$ \emph{bags} and $N$ the set
of nodes. Then, the pair~${\mathcal{T}} = (T,\chi)$ is a \emph{tree
  decomposition (TD)} of~$G$ if the following conditions hold:
\begin{inparaenum}[(i)]
\item %
  for every vertex~$v \in V$ there is a node~$t \in N$ with
  $v \in \chi(t)$;
\item %
  for every edge~$e \in E$ there is a node~$t \in N$ with
  $e \subseteq \chi(t)$; and
\item %
  for any three nodes~$t_1,t_2,t_3\in N$, if $t_2$ lies on the unique
  path from~$t_1$ to~$t_3$, then
  $\chi(t_1)\cap \chi(t_3) \subseteq \chi(t_2)$.
\end{inparaenum}
We call $\max\SB \Card{\chi(t)} - 1 \SM t \in N\SE$ the \emph{width}
of the TD. The \emph{treewidth}~$\tw{G}$ of a graph~$G$ is the
minimum width over all possible TDs of~$G$.
\longversion{%
}%
\longversion{Note that each graph has a trivial TD~$(T,\chi)$
  consisting of the tree~$(\{n\}, \emptyset, n)$ and the mapping
  $\chi(n) = V$.  It is well known that the treewidth of a tree
  is~$1$, and a graph containing a clique of size $k$ has at least
  treewidth $k-1$.}
For some arbitrary but fixed integer~$k$ and a graph of treewidth at
most~$k$, we can compute a TD of width $\leqslant k$ in
time~$2^{\bigO{k^3}} \cdot \Card{V}$~\cite{BodlaenderKoster08}.
Given a TD $(T,\chi)$ with $T = (N,\cdot,\cdot)$, for a node~$t \in N$
we say that $\type(t)$ is $\leaf$ if $t$ has no children; $\join$ if
$t$ has children~$t'$ and $t''$ with $t'\neq t''$ and
$\chi(t) = \chi(t') = \chi(t'')$; $\intr$ (``introduce'') if $t$ has a
single child~$t'$, $\chi(t') \subseteq \chi(t)$ and
$|\chi(t)| = |\chi(t')| + 1$; $\rem$ (``removal'') if $t$ has a single
child~$t'$, $\chi(t') \supseteq \chi(t)$ and
$|\chi(t')| = |\chi(t)| + 1$. If every node $t\in N$ has at most two
children, $\type(t) \in \{ \leaf, \join, \intr, \rem\}$, and bags of
leaf nodes and the root are empty, then the TD is called \emph{nice}.
For every TD, we can compute a nice TD in linear time without
increasing the width~\cite{BodlaenderKoster08}.
Later, we traverse a TD bottom up, therefore, let
$\post(T,t)$ be the sequence of nodes in post-order of the induced
subtree~$T'=(N',\cdot, t)$ of $T$ rooted at~$t$. 

\shortversion{\noindent \textit{Answer Set programming (ASP).} }%
\longversion{\subsection{Answer Set programming (ASP)}}%
\emph{ASP} is a declarative modelling and problem solving framework
that combines techniques of knowledge representation and database
theory. A main advantage of ASP is its expressiveness and when using
non-ground programs the advanced declarative problem modelling
capability. Prior to solving, non-ground programs are usually compiled
into grounded by a grounder. In this paper, we restrict ourselves to
ground ASP programs.
For a comprehensive introduction,
see,~e.g.,~\cite{JanhunenNiemela16a}.
\noindent Let $\ell$, $m$, $n$ be non-negative integers such that
$\ell \leq m \leq n$, $a_1$, $\ldots$, $a_n$ distinct propositional
atoms and %
$l \in \{a_1, \neg a_1\}$.
A \emph{choice rule} is an expression of the form
$\{a_1; \ldots; a_\ell \} \hsep a_{\ell+1}, \ldots, a_m, \neg a_{m+1},
\ldots, \neg a_n$,
\noindent intuitively some subset of $\{a_1, \ldots, a_\ell\}$ is true
if all atoms $a_{\ell+1}, \ldots, a_m$ are true and there is no
evidence that any atom of $a_{m+1}, \ldots, a_n$ is true.
A \emph{disjunctive rule} is of the form
$a_1\por \cdots \por a_\ell \hsep a_{\ell+1}, \ldots, a_{m}, \neg$
$a_{m+1}$, $\ldots$, $\neg a_n$, %
\noindent intuitively at least one atom of $a_1, \ldots, a_\ell$ must
be true if all atoms $a_{\ell+1}, \ldots, a_m$ are true and there is
no evidence that any atom of $a_{m+1}, \ldots, a_n$ is true.
An \emph{optimization rule} is an expression of the form
$\optimize l$. 
with the intuitive meaning that when literal~$l$ is true, this incurs
a penalty of weight~$w$.
\longversion{
} 
A \emph{rule} is either a disjunctive, a choice, %
or an optimization rule.
For a choice or disjunctive %
rule~$r$, let $H_r \eqdef \{a_1, \ldots, a_\ell\}$,
$B^+_r \eqdef \{a_{\ell+1}, \ldots, a_{m}\}$,
and $B^-_r \eqdef \{a_{m+1}, \ldots, a_n\}$.
Usually, if $B^-_r \cup B^+_r = \emptyset$ we write for a rule~$r$
simply~$H_r$ instead of $H_r \hsep$.
For an optimization rule~$r$, %
if $l=a_1$,
let $B^+_r\eqdef \{a_1\}$ and $B^-_r\eqdef \emptyset$; and if
$l=\neg a_1$, let $B^-_r\eqdef \{a_1\}$ and $B^+_r\eqdef \emptyset$.
For a rule $r$, let $\at(r) \eqdef H_r \cup B^+_r \cup B^-_r$ denote
its \emph{atoms} and
$B_r \eqdef B^+_r \cup \SB \neg b \SM b \in B^-_r \SE$ 
its \emph{body}.
Let a \emph{program}~$\prog$ be a set of rules, and $\at(\prog) \eqdef \bigcup_{r\in\prog}\at(r)$ %
denote its atoms. %
and let $\choice(\prog)$, $\disj(\prog)$, and
$\opt(\prog)$ %
denote the set of all choice, disjunctive, and optimization
rules %
in~$\prog$, respectively.
\longversion{%
} 
A set $M \subseteq \at(\prog)$ \emph{satisfies} a rule~$r$ if
(i)~$(H_r \cup B^-_r) \cap M \neq \emptyset$ or $B^+_r \not\subseteq
M$ for~$r\in \disj(\prog)$ or
(ii) $r$ is a choice or optimization rule. %
$M$ is a model of~$\prog$, denoted by $M \models \prog$, if $M$
satisfies every rule~$r \in \prog$. 
\longversion{%
}
The \emph{reduct}~$r^M$  
(i)~of a choice rule~$r$ is the set
$\SB a \leftarrow B^+_r \SM a \in H_r \cap M, B^-_r \cap M =
\emptyset\SE$ of rules, and
(ii)~of a disjunctive rule~$r$ is the singleton
$\SB H_r \leftarrow B^+_r \SM B^-_r \cap M = \emptyset\SE$. %
$\prog^M := \bigcup_{r\in\prog}r^M$ is called \emph{GL
  reduct} of $\prog$ with respect to~$M$.
A set~$M \subseteq \at(\prog)$ is an \emph{answer set} of
$\prog$ if (i) %
$M \models \prog$ and (ii) there is no
$M' \subsetneq M$ such that $M' \models \prog^M$, that is,
$M$ is \emph{subset minimal with respect to $\prog^M$}.
\longversion{%
}%
We call
$\oo(\prog,M)\eqdef\allowbreak\Card{\{r\mid r\in\prog, r \text{ is an optimization rule}, (B^+_r \cap M) \cup (B^-_r \setminus M) \neq \emptyset\}}$
the \emph{cost} of answer set $M$ for $\prog$. An answer set~$M$ of $\prog$ is
\emph{optimal} if its cost is minimal over all answer sets.

\begin{example}%
  \label{ex:running}
  Consider program\\
  $P = \SB %
  \overbrace{\{ e_{ab} \}}^{r_{ab}}\rsep %
  \overbrace{\{ e_{bc} \}}^{r_{bc}}\rsep %
  \overbrace{\{ e_{cd} \}}^{r_{cd}}\rsep %
  \overbrace{\{ e_{ad} \}}^{r_{ad}} \rsep %
  \overbrace{ a_b \hsep e_{ab} }^{r_b}\rsep %
  \overbrace{ a_d \hsep e_{ad} }^{r_d}\rsep %
  \overbrace{ a_c \hsep a_b, e_{bc} }^{r_{c1}}\rsep %
  \overbrace{ a_c \hsep a_d, e_{cd} }^{r_{c2}}\rsep %
  \overbrace{ \hsep \neg a_c }^{r_{\neg c}}%
  \SE.$\\
  The set~$A = \{e_{ab}, e_{bc}, a_b, a_c\}$ is an answer set of~$\prog$, since
  $\{e_{ab}, e_{bc}, a_b, a_c\}$ is the only minimal model of~$P^A = \SB %
  e_{ab} \leftarrow \rsep %
  e_{bc} \leftarrow \rsep %
  a_b \hsep e_{ab} \rsep %
  a_d \hsep e_{ad} \rsep %
  a_c \hsep a_b, e_{bc} \rsep %
  a_c \hsep a_d, e_{cd} \SE.$
  Then, consider program
  $R = \{ a \por c \leftarrow b\rsep b \leftarrow c, \pnot g \rsep c
  \leftarrow a\rsep b \por c \leftarrow e \rsep h \por i \leftarrow g,
  \pnot c \rsep a \por b \rsep g \leftarrow \pnot i \rsep c \rsep
  \{d\} \leftarrow g \}$.  The set $B =\{ b, c, d, g \}$ is an answer
  set of $R$ since $\{b,c,d, g\}$ and $\{a,c,d,g\}$ are the minimal
  models of
  $R^B= \SB a \por c \leftarrow b \rsep c \leftarrow a\rsep b \por c
  \leftarrow e \rsep a \por b \rsep g \rsep c \rsep d \leftarrow g
  \SE$.
\end{example}

\longversion{%
}%
\noindent Given a program~$\prog$, we consider the problems of computing an answer
set (called \AspComp), outputting the number of optimal answer
sets (called \AspCountO), 
and listing all optimal answer sets of~$\prog$ (called
\AspEnum). Further, given a propositional formula~$F$ and an
atom~$sol$, we use the entailment problem of listing every
subset-minimal model~$M$ of $F$ with $sol\in M$ (called \MSatEnum).

\shortversion{\smallskip \noindent \textit{Graph Representations of Programs.} }%
\longversion{\subsection{Graph Representations of Programs}}%
In order to use TDs for ASP solving, we need dedicated graph
representations of programs.
The \emph{incidence graph}~$I(\prog)$ of $\prog$ is the bipartite
graph that has the atoms and rules of~$\prog$ as vertices and an
edge~$a\, r$ if $a \in \at(r)$ for some
rule~$r \in \prog$~\cite{FichteEtAl17}.
The \emph{semi-incidence graph}~$S(\prog)$ of $\prog$ is a graph that
has the atoms and rules of~$\prog$ as vertices and (i)~an edge~$a\, r$
if $a \in \at(r)$ for some rule~$r \in \prog$ as well as (ii)~an
edge~$a\, b$ for disjoint atoms~$a,b \in H_r$
where~$r\in\prog$ is a choice rule. %
Since for every program~$P$ the incidence
graph~$I(P)$ is a subgraph of the semi-incidence graph, we have that
$\tw{I(\prog)} \leq \tw{S(\prog)}$.
Further, by definition of a TD and the construction of a
semi-incidence graph that head atoms of choice rules, respectively,
occur in at least one common bag of the TD.

 \shortversion{\smallskip \noindent \textit{Sub-Programs.} }%
 \longversion{\subsection{Sub-Programs}}%
 Let ${\cal T} = (T, \chi)$ be a nice TD of graph
 representation~$S(\prog)$ of a program
 $\prog$. Further, let $T = (N,\cdot,n)$ and $t \in N$.
 The \emph{bag-program} is defined as
 $\prog_t \eqdef \prog \cap \chi(t)$.
 Further, the set~$\atto \eqdef \SB a \SM a \in \at(\prog) \cap \chi(t'), t' \in
 \post(T,t) \SE$ is called \emph{atoms below~$t$}, the \emph{program below $t$}
 is defined as $\progt{t} \eqdef \SB r \SM r \in \prog_{t'}, t' \in \post(T,t) \SE$,
 and the \emph{program strictly below $t$} is $\progtneq{t}\eqdef
 \progt{t}\setminus \prog_t$. It holds that $\progt{n} = \progtneq{n} = \prog$ and
 $\att{n}=\at(\prog)$. 
\section{A Single Pass DP Algorithm}
\label{sec:algo:dp}

A dynamic programming based \ASP solver, such as
DynASP2~\cite{FichteEtAl17}, splits the input program~$\prog$ into
``bag-programs'' based on the structure of a given nice tree
decomposition for~$\prog$ and evaluates~$\prog$ in parts, thereby
storing the results in tables for each TD node.  More precisely, the
algorithm works as outlined on the left and middle of Figure~\ref{fig:dp-approach} and encompasses the following steps:
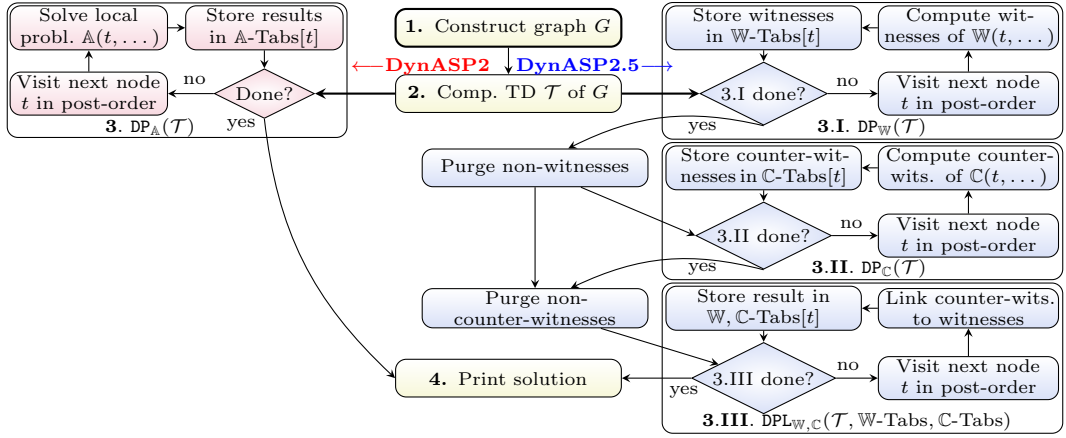
\begin{figure}[t]
  \centering %
    \begin{tikzpicture}[
	rounded corners,
	rect/.style={
		draw,
		rectangle,
		minimum height=5mm,
		top color=white,
		middle color=white,
	},
	dflatrect/.style={
		rect,
		bottom color=red!80!blue!15,
	},
	asprect/.style={
		rect,
		bottom color=blue!80!green!15,
	},
	decomprect/.style={
		rect,
		bottom color=yellow!80!black!15,
	},
	diam/.style={
		draw,
		diamond,
		shape aspect=2,
		sharp corners,
		inner color=white,
		outer color=red!80!blue!15,
	},
	diam25/.style={
		draw,
		diamond,
		shape aspect=2,
		sharp corners,
		inner color=white,
		outer color=blue!80!green!15,
	},
	font=\scriptsize,
	>=stealth, inner sep=1
]

\matrix[column sep=2mm, row sep=1.5mm]{

\node [dflatrect, text width=20mm, text centered] (solve1) {Solve local probl.\ ${\AlgA}(t,\dots)$};
&\node [dflatrect, text width=20mm, text centered] (populate1) {Store results in $\Tabs{t}$ }; &
\node[xshift=2mm]{};&
\node [decomprect, xshift=2.5mm, text centered, text width=29mm, thick] (parse) {\textbf{1.} Construct graph~$G$};
&
\node [asprect, text width=25mm,text centered] (populate) {Store witnesses in $\ATabs{{\AlgW}}{t}$}; &
\node [asprect, text width=23mm, text centered] (solve) {Compute witnesses of ${\AlgW}(t,\dots)$};
\\
\node [dflatrect, text width=20mm, text centered] (traverse1) {Visit next node $t$ in post-order};
&\node [diam] (done1) {Done?};
\node [above left, xshift=-1mm, yshift=1.25mm, inner sep=0] at (done1.west) {no};
\node [below right, xshift=-5mm, yshift=0mm, inner sep=0] at (done1.south) {yes};
\node[]{};&
&\node [decomprect,xshift=2.5mm, text width=29mm, text centered] (decompose) {\textbf{2.} Comp.\ TD~${\cal T}$ of~$G$}; 
&
\node [diam25, inner sep=1] (done) {3.I done?};
\node [above right, xshift=1mm, yshift=1mm, inner sep=0] at (done.east) {no};
\node [below right, xshift=-10mm, yshift=0.5mm, inner sep=0] at (done.south) {yes};
&
\node [asprect, text width=23mm, text centered] (traverse) {Visit next node $t$ in post-order};
\\
&&
\node[]{};&
\node [asprect, text width=29mm, xshift=6mm, text centered] (clean) {Purge non-witnesses};&
\node [asprect, text width=25mm, text centered] (populate2) {Store counter-wit-nesses\hspace{0.17em}in\hspace{0.2em}$\ATabs{{\AlgC}}{t}$}; &
\node [asprect, text width=23mm, text centered] (solve2) {Compute counter-wits. of ${\AlgC}(t, \dots)$};
\\
&&
\node[]{};&
&\node [diam25, inner sep=1] (done2) {3.II done?};
\node [above right, xshift=1mm, yshift=1mm, inner sep=0] at (done2.east) {no};
\node [below right, xshift=-10mm, yshift=1mm, inner sep=0] at (done2.south) {yes};
&
\node [asprect, text width=23mm, text centered] (traverse2) {Visit next node $t$ in post-order};
\\
&&
\node[]{};&
\node [asprect, text width=29mm, xshift=6mm,text centered] (clean2) {Purge non-counter-witnesses};
&
\node [asprect, text width=25mm,text centered] (populate3) {Store result in $\ATabs{{\AlgW},{\AlgC}}{t}$};
&
\node [asprect, text width=23mm, text centered] (solve3) {Link counter-wits. to witnesses};
\\
&&
\node[]{};&
\node [decomprect, xshift=2.5mm,text width=29mm, text centered] (reconstruct) {\textbf{4.} Print solution};
&
\node [diam25, inner sep=1] (done3) {3.III done?};
\node [above right, xshift=0mm, yshift=1mm, inner sep=0] at (done3.east) {no};
\node [below left, xshift=1mm, yshift=-1mm, inner sep=0] at (done3.west) {\;yes};
&
\node [asprect, text width=23mm, text centered] (traverse3) {Visit next node $t$ in post-order};
\\};
\node [below=of populate,xshift=1.3cm,yshift=1.64cm,text width=52mm, text height=16.8mm, draw, align=right] {$\mathbf{3.I.}~ \dpa_{\AlgW}({\cal T})$\qquad\qquad\qquad\;\vspace{-0.00em}};
\node [below=of populate2,xshift=1.3cm,yshift=1.64cm,text width=52mm, text height=17mm, draw, align=right] {$\mathbf{3.II.}~ \dpa_{\AlgC}({\cal T})$\qquad\qquad\qquad\;\vspace{-0.05em}};
\node [below=of populate3,xshift=1.3cm,yshift=1.64cm,text width=52mm, text height=18.65mm, draw, align=right] {$\mathbf{3.III.}~ \cw_{\AlgW, {\AlgC}}({\cal T}, \ATab{\AlgW}, \ATab{\AlgC})$\qquad\;\vspace{-0.05em}};
\node [xshift=-4.70cm,yshift=1.85cm,text width=44mm, text height=16.8mm, draw, align=right] {$\mathbf{3.}~ \dpa_{\AlgA}({\cal T})$\qquad\qquad\qquad\quad\vspace{-0.10em}};
\node [xshift=-1.5cm,yshift=1.9cm] {\textcolor{red}{$\longleftarrow$\textbf{DynASP2}}};
\node [xshift=0.8cm,yshift=1.9cm] {\textcolor{blue}{\textbf{DynASP2.5}$\longrightarrow$}};

\draw (parse) edge[->] (decompose);
\draw (decompose) edge[->,thick] (done);
\draw (decompose) edge[->,thick] (done1);
\draw (done) edge[->] (traverse);
\draw (done.south) edge[->, bend right=20]  (clean);
\draw (clean) edge[->] (clean2);
\draw (done3.west) edge[->] (reconstruct);
\draw (traverse) edge[->] (solve);
\draw (solve) edge[->] (populate);
\draw (populate) edge[->] (done);

\draw (populate2) edge[->] (done2);
\draw (solve2) edge[->] (populate2);
\draw (traverse2) edge[->] (solve2);
\draw (done2) edge[->] (traverse2);
\draw (done2.south) edge[->,bend right=20] (clean2);
\draw (clean) edge[->] (done2.west);

\draw (done3) edge[->] (traverse3);
\draw (populate3) edge[->] (done3);
\draw (solve3) edge[->] (populate3);
\draw (traverse3) edge[->] (solve3);

\draw (clean2) edge[->] (done3);

\draw (solve1) edge[->] (populate1);
\draw (populate1) edge[->] (done1);
\draw (done1) edge[->] (traverse1);
\draw (traverse1) edge[->] (solve1);
\draw (done1.south) edge[->, bend right=20] (reconstruct.west);
\end{tikzpicture}

    \caption{Control flow for DP-based \ASP solver (DynASP2, left) and for {\dynaspplus{\cdot}} (right).}
\label{fig:dp-approach}
\label{fig:dp-advanced-approach}

\end{figure}

\begin{algorithm}[t]
  \KwData{Table algorithm $\AlgA$, nice TD~$\TTT=(T,\chi)$ with
    $T=(N,\cdot,n)$ of $G(\prog)$ according to $\AlgA$.}%
  \KwResult{$\ATab{\AlgA}$: maps each TD node~$t\in T$ to some computed
    table~$\tau_t$. } %
  $\Tab{}_t \eqdef \SB \ATab{\AlgA}[t'] \SM t' \text{ is a child of $t$ in
      $T$}\SE$\;\vspace{0.2em} %
  \For{\text{\normalfont iterate} $t$ in \text{\normalfont post-order}(T,n)}{\vspace{-0.05em}%
    %
    $\ATab{\AlgA}[t] \leftarrow {\AlgA}(t,\chi(t),\prog_t,\atto,\Tab{}_t)$\; %
    \vspace{-0.5em} }\vspace{-0.2em}%
  \caption{Algorithm ${\dpa}_{\AlgA}({\cal T})$ for Dynamic
    Programming on TD ${\cal T}$ for ASP~\protect\cite{FichteEtAl17}.}
\label{fig:dpontd}
\end{algorithm}
%

\begin{enumerate}
\item Construct a graph representation~$G(\prog)$ of the given input
  program~$\prog$.
\item Compute a TD~$\TTT$ of the graph~$G(\prog)$ by means of some heuristic, thereby
  decomposing~$G(P)$ into several smaller parts and fixing an
  ordering in which~$\prog$ will be evaluated.
\item\label{step:td:solving} Algorithm~\ref{fig:dpontd} --
  ${\dpa}_{\AlgA}({\cal T})$ -- sketches the general scheme for this
  step, assuming that an algorithm~${\AlgA}$, which highly depends on
  the graph representation, is given. We usually call $\AlgA$ the
  \emph{table algorithm}\footnote{The table algorithm~$\INC$ for example is given in
  Algorithm~\ref{fig:incinc}.}. %
  For every node~$t \in T$ in the tree
  decomposition~$\TTT = ((T,E,n),\chi)$ (in a bottom-up traversal),
  run~${\AlgA}$ and compute~$\Tabs{t}$, which are sets of tuples (or
  \emph{rows} for short).  Intuitively, algorithm~$\AlgA$ transforms tables of child nodes of~$t$ to the current node, 
  and solves a
  ``local problem'' using bag-program~$\prog_t$. The algorithm thereby
  computes %
  (i)~sets of atoms called (local) \emph{witness
    sets} %
  and (ii)~for each local witness set~$M$ subsets of~$M$ called
  \emph{counter-witness sets}~\cite{FichteEtAl17}, and directly follows the definition of answer sets being (i)~models of
  $\prog$ and (ii)~subset minimal with respect to~$\prog^M$.
\item %
  For root~$n$ interpret the table~$\Tabs{n}$ (and tables of children,
  if necessary) and print the solution to the considered \ASP
  problem.%
\end{enumerate}

\noindent Next, we propose a new table algorithm~($\INC$) for programs
without optimization rules. Since our algorithm trivially extends to
counting and optimization rules by earlier work~\cite{FichteEtAl17},
we omit such rules. The table algorithm~${\INC}$ employs the
semi-incidence graph and is depicted in Algorithm~\ref{fig:incinc}.
$\dpa_{\INC}$ merges two earlier algorithms for the primal and
incidence graph~\cite{FichteEtAl17} resulting in slightly different
worst case runtime bounds (c.f., Theorem~\ref{thm:inc:runtime}).
Our table algorithm~\INC computes and stores (i)~sets of atoms
(witnesses) that are relevant for the \SAT part (finding a model of
the program) and (ii)~sets of atoms (counter-witnesses) that are
relevant for the \UNSAT part (checking for minimality). In addition,
we need to store for each set of witnesses as well as its set of
counter-witnesses satisfiability states (\emph{sat-states} for short).
For the following reason: 
By Definition of TDs and the semi-incidence graph, it is true for
every atom~$a$ and every rule~$r$ of a program that if atom~$a$ occurs
in rule~$r$, then $a$ and $r$ occur together in at least one bag of
the TD. In consequence, the table algorithm encounters every
occurrence of an atom in any rule. In the end, on removal of~$r$, we
have to ensure that $r$ is among the rules that are already
satisfied. However, we need to keep track whether a witness
\emph{satisfies} a rule, because not all atoms that occur in a rule
occur together in exactly one bag. Hence, when our algorithm traverses
the TD and an atom is forgotten we still need to store this
\emph{sat-state}, as setting the forgotten atom to a certain truth
value influences the satisfiability of the rule.
Since the semi-incidence graph contains a clique on every set~$A$ of
atoms that occur together in %
choice rule head, %
those atoms~$A$ occur together in a common bag of any TD of the
semi-incidence graph. For that reason, we do \emph{not} need to
incorporate %
choice rules into the satisfiability state, in
contrast to the algorithm for the incidence graph~\cite{FichteEtAl17}.
We can see witness sets together with its sat-state as \emph{witness}.
Then, in Algorithm~\ref{fig:incinc} (\INC) a row in the
table~$\tab{t}$ is a triple~$\langle M, \sigma, \CCC \rangle$.  The
set~$M \subseteq \at(\prog)\cap\chi(t)$ represents a witness set. 
The family~$\CCC$ of sets concerns counter-witnesses, which we will
discuss in more detail below.
The sat-state~$\sigma$ for $M$ represents rules of $\chi(t)$ satisfied
by a superset of~$M$.  Hence, $M$ witnesses a model~$M'\supseteq M$
where $M' \models \progtneq{t} \cup \sigma$.  
We use binary operator~$\cup$ to combine sat-states, which ensures
that rules satisfied in at least one operand remain satisfied.
We compute a new sat-state~$\sigma$ from a sat-state and satisfied
rules, formally,
$\SSR(\dot{R},M)\eqdef \{r \mid (r,{R})\in\dot{R}, M
\models {R} \}$ for $M \subseteq \chi(t) \setminus \prog_t$ and
program~$\dot{R}(r)$ constructed by~$\dot{R}$, mapping rules
to local-programs (Definition~\ref{def:bagprogram}).
\begin{definition}\label{def:bagprogram}%
  Let $\prog$ be a program, $\TTT=(\cdot,\chi)$ be a TD of $S(\prog)$,
  $t$ be a node of $\TTT$ and ${R} \subseteq \prog_t$.
  The \emph{local-program}~${R}^{(t)}$ is obtained
  from~${R} \cup \{ \leftarrow B_r \mid r\in R \text{ is a choice rule, } %
  H_r \subsetneq \att{t}\}$\shortversion{\footnote{%
      We require to add
      $\{\leftarrow B_r \mid r\in R \text{ is a choice rule, } H_r \subsetneq
      \att{t}\}$ in order to decide satisfiability for corner cases of
      choice rules involving counter-witnesses of Line~3 in
      Algorithm~\ref{fig:incinc}.}}%
  \longversion{\footnote{%
      We require to add
      $\{\leftarrow B_r \mid r\in R \text{ is a choice rule, } H_r \subsetneq
      \att{t}\}$ in order to decide satisfiability for corner cases of
      choice rules involving counter-witnesses of Line~3 in
      Algorithm~\ref{fig:incinc}.}}\footnoteitext{\label{foot:sigma}\label{foot:abrevtwo}%
  For set $S$ and element $s$, we denote
  $\MAIR{S}{s}\hspace{-0.15em}\eqdef\hspace{-0.1em}S \cup \{s\}$ and
  $\MARR{S}{s}\hspace{-0.15em}\eqdef\hspace{-0.1em}S\setminus\{s\}$.%
} %
  by %
removing from every rule all
    literals~$a, \neg a$ with $a \not\in \chi(t)$. %
  \shortversion{}%
  We
  define~$\dot{R}^{(t)}: {R} \rightarrow 2^{{R}^{(t)}}$ by
  $\dot{R}^{(t)}(r) \eqdef \{r\}^{(t)}$ for
  $r\in {R}$.
\end{definition}%
\begin{example}
  Observe $\prog_{t_4}^{(t_4)} = \{\leftarrow e_{bc}, r_b\}$ and
  $\prog_{t_5}^{(t_5)} = \{c \leftarrow\}$ for $\prog_{t_4}$ and
  $\prog_{t_5}$ of Figure~\ref{fig:running1_inc}. %
\end{example}
\renewcommand{\eqdef}{\leftarrow}
\begin{algorithm}[t]
   \KwData{Bag $\chi_t$, bag-program $\prog_t$, atoms-below $\atto$, child tables $\Tab{}_t$ of $t$.{~\bf Out:} Tab.~$\tab{t}$.} 
   %

   %
   %
   \lIf(\tcc*[f]{\hspace{-0.05em}Abbreviations see Footnote~\ref{foot:sigma}.\hspace{-0.05em}}){$\type(t) = \leaf$}{%
     $\tab{t} \eqdef \Big\{ \Big\langle
     \tuplecolor{\inputPredColor}{\emptyset},
     \tuplecolor{\statePredColor}{\emptyfunc},
     ~\tuplecolor{\outputPredColor}{\emptyset}\trash{, \tuplecolor{\specialPredColor}{\{()\}}} \Big\rangle \Big\}$
   \vspace{-0.00em}}%
   %
   %
   %
   \uElseIf{$\type(t) = \intr$, $a \in \chi_t \setminus \prog_t$ is introduced and $\tau'\in \Tab{}_t$}{     
     \vspace{-0.25em}$\hspace{-0.9em}\tab{t} \eqdef \Big\{ \Big\langle \tuplecolor{\inputPredColor}{\MAI{M}}, \tuplecolor{\statePredColor}{\sigma \cup
     \SSR(\dot\prog_t^{(t)},\MAI{M}) },~\tuplecolor{\outputPredColor}{\{ \langle\MAI{C},
       \rho \cup \SSR(\dot\prog_t^{(t,{\MAI{M}})},\MAI{C})\rangle \mid \langle C,
       \rho\rangle \in \CCC \}~\cup}$

  \vspace{-0.15em}%
       \makebox[0.25cm][l]{}\makebox[10.80cm][l]{$\tuplecolor{\outputPredColor}{\{ \langle C,
       \rho \cup \SSR(\dot\prog_t^{(t,{\MAI{M}})},C)\rangle \mid \langle C, \rho
       \rangle \in \CCC \}~\cup~}\tuplecolor{\outputPredColor}{\{\langle M,
       \sigma \cup \SSR(\dot\prog_t^{(t,{\MAI{M}})},M)\rangle \}}\trash{, \tuplecolor{\specialPredColor}{\{(\tabval)\}}}\Big\rangle$}%
	   $\Bigm|\;\trash{\tabval=}\langle \tuplecolor{\inputPredColor}{M}, \tuplecolor{\statePredColor}{\sigma}, \tuplecolor{\outputPredColor}{\CCC}\trash{, \tuplecolor{\specialPredColor}{\cdot}} \rangle \in \tab{}'\Big\}$

\vspace{-0.10em}%
	
	\makebox[0.41cm][l]{$\;\cup$}\makebox[10.64cm][l]{$\Big\{ \Big\langle \tuplecolor{\inputPredColor}{M},
        \tuplecolor{\statePredColor}{\sigma \cup
          \SSR(\dot\prog_t^{(t)},M)},~
        \tuplecolor{\outputPredColor}{\{ \langle C, \rho \cup
          \SSR(\dot\prog_t^{(t,M)},C) \rangle \mid \langle C,
          \rho\rangle \in \CCC \}}$%
        $\trash{,\tuplecolor{\specialPredColor}{\{(\tabval)\}}}\Big\rangle$}%
      $\Bigm|\;\trash{\tabval = }\langle
      \tuplecolor{\inputPredColor}{M},
      \tuplecolor{\statePredColor}{\sigma},
      \tuplecolor{\outputPredColor}{\CCC}\trash{,
        \tuplecolor{\specialPredColor}{\cdot}} \rangle \in \tab{}'
      \Big\}$}

     %
   %
   %
   \uElseIf{$\type(t) = \intr$, $r \in \chi_t \cap \prog_t$ is introduced and $\tau'\in \Tab{}_t$}{
     $\hspace{-0.9em}\tab{t} \eqdef$ \makebox[10.553cm][l]{$\Big\{ \Big\langle \tuplecolor{\inputPredColor}{M}, \tuplecolor{\statePredColor}{{\sigma} \cup
     \SSR(\{\dot r\}^{(t)},M)},~\tuplecolor{\outputPredColor}{\{ \langle C, {\rho} \cup
       \SSR(\{\dot r\}^{(t,M)},C) \rangle \mid \langle C, \rho \rangle \in \CCC\}}\trash{, \tuplecolor{\specialPredColor}{\{(\tabval)\}}} \Big\rangle$}%
	   $\Bigm|\; \trash{\tabval=}\langle \tuplecolor{\inputPredColor}{M}, \tuplecolor{\statePredColor}{\sigma}, \tuplecolor{\outputPredColor}{\CCC} \rangle \in \tab{}' \Big\}$
 
     %
   \vspace{-0.10em}}
   %
   %
   %
   \uElseIf{$\type(t) = \rem$, $a \not\in \chi_t$ is removed atom and $\tau'\in \Tab{}_t$}{
       $\hspace{-0.9em}\tab{t} \eqdef$ \makebox[10.553cm][l]{$\trash{\compr(}\Big\{ \Big\langle \tuplecolor{\inputPredColor}{\MAR{M}},\tuplecolor{\statePredColor}{\sigma
       },~\tuplecolor{\outputPredColor}{\{ \langle \MAR{C}, \rho \rangle
       \mid \langle C, \rho \rangle \in \CCC \}}\trash{, \tuplecolor{\specialPredColor}{\{(\tabval)\}}} \Big\rangle$}%
       $\Bigm|\;\trash{\tabval =} \langle \tuplecolor{\inputPredColor}{M}, \tuplecolor{\statePredColor}{\sigma}, \tuplecolor{\outputPredColor}{\CCC}\trash{, \tuplecolor{\specialPredColor}{\cdot}} \rangle \in \tab{}'\Big\}\trash{)}$
       %
%
     %
   }%
   %
   %
   \vspace{-0.07em}
   \uElseIf{$\type(t) = \rem$, $r \not\in \chi_t$ is removed rule and $\tau'\in \Tab{}_t$}{
     $\hspace{-0.9em}\tab{t} \eqdef$ \makebox[9.66cm][l]{$\Big\{\Big\langle \tuplecolor{\inputPredColor}{M},
       \tuplecolor{\statePredColor}{\MARR{\sigma}{r}},$~%
       $\tuplecolor{\outputPredColor}{\big\{ \langle C, \MARR{\rho}{r} \rangle \mid \langle C, \rho
       \rangle \in \CCC, r\in \rho  
       \big\}}\trash{, \tuplecolor{\specialPredColor}{\{(\tabval)\}}} \Big\rangle$}%
%
     $\Bigm|\;\trash{\tabval =} \langle \tuplecolor{\inputPredColor}{M}, \tuplecolor{\statePredColor}{\sigma}, \tuplecolor{\outputPredColor}{\CCC}\trash{, \tuplecolor{\specialPredColor}{\cdot}} \rangle \in \tab{}', r\in\sigma 
     \hspace{-0.015em} \Big\}$ 
	 }\vspace{-0.05em}

   %
   %
   \uElseIf{$\type(t) = \join$ and $\tau', \tau'' \in \Tab{}_t$ with $\tau' \neq \tau''$}{
     $\hspace{-0.9em}\tab{t} \eqdef\trash{\compr(}\Big\{ \Big\langle \tuplecolor{\inputPredColor}{M}, \tuplecolor{\statePredColor}{\sigma' \cup \sigma''},~\tuplecolor{\outputPredColor}{\{
     \langle C, \rho' \cup \rho''\rangle \mid \langle C, \rho'
     \rangle \in \CCC', \langle C, \rho'' \rangle \in \CCC''\}~\cup~\tuplecolor{\outputPredColor}{\{ \langle M,
     \rho \cup \sigma''\rangle \mid \langle M, \rho \rangle \in \CCC'
     \}~\cup}}$%
  
   \makebox[0.55cm][l]{}%
    %
   \makebox[7.79cm][l]{$\tuplecolor{\outputPredColor}{\{ \langle M, \sigma' \cup \rho
     \rangle \mid \langle M, \rho \rangle \in \CCC''
     \}}\trash{,}\Big\rangle$}%
   $\Bigm|\;\trash{\tabval'=}\langle
   \tuplecolor{\inputPredColor}{M},
   \tuplecolor{\statePredColor}{\sigma'},
   \tuplecolor{\outputPredColor}{\CCC'}\trash{,
     \tuplecolor{\specialPredColor}{\cdot}} \rangle\in\tab{}',
   \trash{\tabval''=}\langle \tuplecolor{\inputPredColor}{M},
   \tuplecolor{\statePredColor}{\sigma''},
   \tuplecolor{\outputPredColor}{\CCC''}\trash{,
     \tuplecolor{\specialPredColor}{\cdot}}
   \rangle\in\tab{}''\trash{, \tab{}' \prec
     \tab{}''}\Big\}\trash{)}$ 
 } \vspace{-0.3em}
\caption{Table algorithm~$\algo{SINC}(t,\chi_t,\prog_t,\atto,\Tab{}_t)$.}
\label{fig:incinc}
\end{algorithm}
%
\renewcommand{\eqdef}{{\ensuremath{\,\mathrel{\mathop:}=}}}
In Example~\ref{ex:sinc:main} we give an idea how we compute models of
a given program using the semi-incidence graph.  The resulting
algorithm \INCSAT is obtained from \INC, by taking only the first two
row positions (red and green parts).
The remaining position (blue part), can be seen as an
algorithm~(\INCCOUNTERSAT) that computes counter-witnesses (see
Example~\ref{ex:running_counter}).
Note that we discuss selected cases, and we assume row numbers in each
table~$\tab{t}$,~i.e., the $i\,^{th}$-row corresponds to
$\tabval_{t.i} = \langle M_{t.i}, \sigma_{t.i} \rangle$.

\begin{example}\label{ex:sinc:main}
  Consider program~$\prog$ from Example~\ref{ex:running},
  TD~$\TTT=(\cdot, \chi)$ in Figure~\ref{fig:running1_inc}, and the
  tables~$\tab{1}$,\nolinebreak $\ldots$, $\tab{34}$, which illustrate
  computation results obtained during post-order traversal of
  ${\cal T}$ by $\dpa_{\PRIMSAT}$. \longversion{Note that}
  Figure~\ref{fig:running1_inc} (left) does not show every
  intermediate node of TD~${\cal T}$.
  Table $\tab{1}= \{\langle \emptyset, \emptyfunc \rangle \}$ as
  $\type(t_1)=\leaf$~(see Algorithm~\ref{fig:incinc} L1).  Table
  $\tab{3}$ is obtained via introducing rule~$r_{ab}$, after
  introducing atom~$e_{ab}$ %
  ($\type(t_2) = \type(t_3) =\intr$). %
  It contains two rows due to two possible truth assignments using
  atom~$e_{ab}$ (L3--5).  Observe that rule~$r_{ab}$ is satisfied in
  both rows $M_{3.1}$ and $M_{3.2}$, since the head of choice
  rule~$r_{ab}$ is in~$\att{t_3}$ (see L7 and
  Definition~\ref{def:bagprogram}).  Intuitively, whenever a rule~$r$
  is proven to be satisfiable, sat-state~$\sigma_{t.i}$ marks~$r$
  \emph{satisfiable} since an atom of a rule of~$S(\prog)$ might only
  occur in one TD bag.
  Consider table~$\tab{4}$ with $\type(t_4)=\rem$ and
  $r_{ab}\in\chi(t_{3})\setminus\chi(t_{4})$.  By
  definition\shortversion{,} \longversion{(TDs and semi-incidence
    graph),} we have encountered every occurrence of any atom
  in~$r_{ab}$. \shortversion{Thus, }\longversion{In consequence,}
  \INCSAT enforces that only rows where $r_{ab}$ is marked satisfiable
  in~$\tab{3}$, are considered for table~$\tab{4}$.  The resulting
  table~$\tab{4}$ consists of rows of~$\tab{3}$ with
  $\sigma_{4.i}=\emptyset$, where rule~$r_{ab}$ is proven satisfied
  ($r_{ab}\in\sigma_{3.1}, \sigma_{3.2}$, see L~11).  Note that
  between nodes~$t_6$ and~$t_{10}$, an atom and rule remove as well as
  an atom and rule introduce node is placed. Observe that the second
  row~$\tabval_{6.2}=\langle M_{6.2},\sigma_{6.2}\rangle \in\tab{6}$
  does not have a ``successor row'' in~$\tab{10}$,
  since~$r_b\not\in\sigma_{6.2}$.
  Intuitively, join node~$t_{34}$ joins only common witness sets
  in~$\tab{17}$ and~$\tab{33}$
  with~$\chi(t_{17})=\chi(t_{33})=\chi(t_{34})$. In general, a join
  node marks rules satisfied, which are marked satisfied in at least
  one child (see L13--14).
\end{example}
\begin{figure}[t]%
\centering %
\vspace{-0.2em}
\includegraphics[scale=1.20]{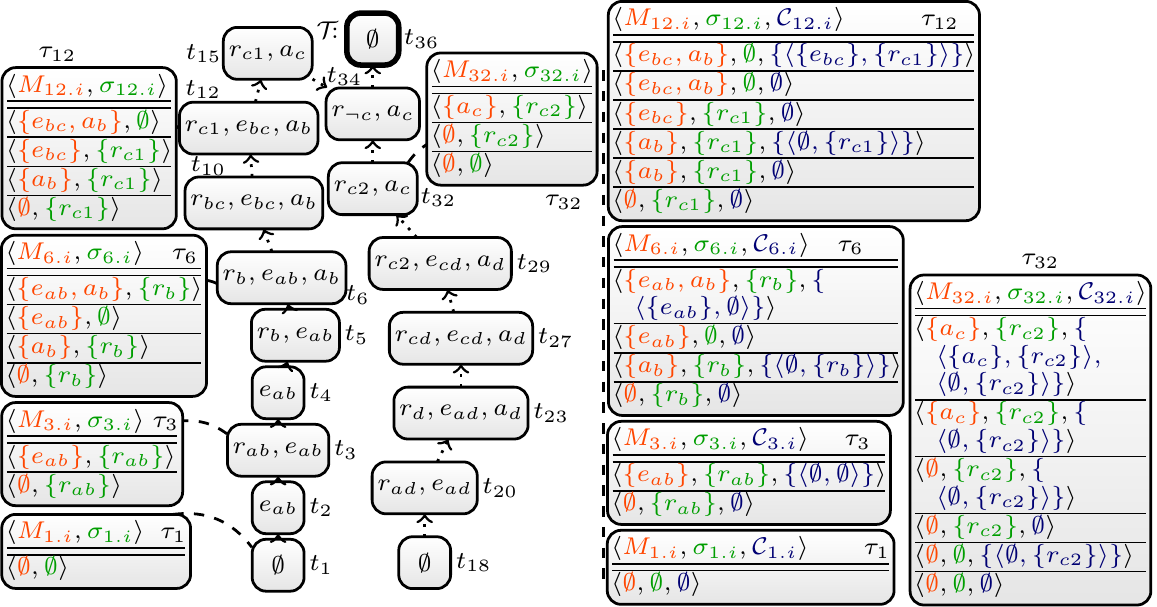}%
\caption{A TD~$\TTT$ of the semi-incidence graph~$S(P)$ for
  program~$P$ from Example~\ref{ex:running} (center).  Selected DP
  tables after~$\dpa_\INCSAT$~(left) and after~$\dpa_\INC$~(right) for
  nice TD~$\TTT$.}
\label{fig:running1_inc}
\end{figure}%
Since we already explained how to obtain models, we only briefly
describe how to compute counter-witnesses. 
Family~$\CCC$ consists of rows~$(C, \rho)$ where
$C\subseteq \at(\prog) \cap \chi(t)$ is a \emph{counter-witness set}
in~$t$ to~$M$.
Similar to the sat-state~$\sigma$, the sat-state~$\rho$ for $C$ under
$M$ represents whether rules of the GL reduct~$\prog_t^M$ are
satisfied by a superset of $C$. We can see counter-witness sets
together with its sat-state as counter-witnesses.  Thus, $C$ witnesses
the existence of $C' \subsetneq M'$ satisfying
$C' \models (\progtneq{t} \cup \rho)^{M'}$ since $M$ witnesses a
model~$M'\supseteq M$ where $M' \models \progtneq{t}$.
In consequence, there exists an answer set of $\prog$ if the root
table contains $\langle \emptyset, \emptyfunc, \emptyset \rangle$.
We require local-reducts for deciding satisfiability of counter-witness sets.

\begin{definition}\label{def:bagreduct}%
  Let $\prog$ be a program, $\TTT=(\cdot,\chi)$ be a TD of $S(\prog)$,
  $t$ be a node of $\TTT$, ${R} \subseteq \prog_t$ %
  and $M\subseteq\at(\prog)$. We define
  \emph{local-reduct~${R}^{(t,M)}$} by
  ${[{R}^{(t)}]}^M$ and  %
  $\dot{R}^{(t,M)}: {R} \rightarrow 2^{{R}^{(t,M)}}$ by
  $\dot{R}^{(t,M)}(r) \eqdef \{r\}^{(t,M)}, r\in R$.
\end{definition}%

\begin{proposition}[c.f.~\cite{FichteEtAl17}]\label{thm:inc:runtime}
  Let $\prog$ be a program and $k\eqdef\tw{S(\prog)}$. Then, the
  algorithm~${\dpa}_{\INC}$ is correct and runs in time
  $\bigO{2^{2^{k + 2}}\cdot \CCard{S(\prog)}}$.%
\end{proposition}

\section{DynASP2.5: Implementing a III Pass DP Algorithm}\label{sec:impl25}

The classical DP algorithm~${\dpa}_{\INC}$ (Step~3 of
Figure~\ref{fig:dp-approach}) follows a single pass approach. It
computes both witnesses and counter-witnesses by traversing the given
TD exactly once. In particular, it stores exhaustively all potential
counter-witnesses, even those counter-witnesses where the witnesses in
the table of a node cannot be extended in the parent node. In
addition, there can be a high number of duplicates among the
counter-witnesses, which are stored repeatedly. %
In this section, we propose a multi-pass approach (\MDP) for DP on TDs
and a new implementation (DynASP2.5), which fruitfully adapts and
extends ideas from a different domain~\cite{BliemEtAl16b}.
Our novel algorithm allows for an early cleanup (purging) of witnesses that do
not lead to answer sets, which in consequence (i)~avoids to construct
expendable counter-witnesses. Moreover, multiple passes enable us to
store witnesses and counter-witnesses separately, which in turn
(ii)~avoids storing counter-witnesses duplicately and (iii)~allows for
highly space efficient data structures (pointers) in practice when
linking witnesses and counter-witnesses together.
Figure~\ref{fig:dp-advanced-approach} (right, middle) presents the control flow
of the new multi-pass approach~\emph{{\dynaspplus{\cdot}}}, 
where \MDP introduces a much more elaborate computation in Step~3.

\subsection{The Algorithm}



\begin{algorithm}[t]
  \KwData{
    Nice TD~$\TTT=(T,\chi)$ with
    $T=(N,\cdot,n)$ of a graph~$S(\prog)$, and
    mappings $\ATab{\AlgW}[\cdot]$, 
    $\ATab{\AlgC}[\cdot]$. 
  }%
  \KwResult{$\ATab{\AlgW, \AlgC}$: maps node~$t\in
    T$ to some pair~$( \tab{t}^\AlgW, \tab{t}^\AlgC
    )$ with $\tab{t}^\AlgW \in \ATabs{\AlgW}{t}, \tab{t}^\AlgC \in \ATabs{\AlgC}{t}$. } %
  
  $\Tab{}_t \eqdef \SB \ATab{\AlgW, \AlgC}[t'] \SM t' \text{ is a child of $t$ in
      $T$}\SE$\;\vspace{-0.05em} %
  
  \tcc*[h]{Get for a node~$t$ tables of (preceeding) combined child rows (CCR)\vspace{-0.05em}}\;
    \vspace{-0.1em}$\label{algo3:ccr}\text{CCR}_t \eqdef \hat{\Pi}_{\tab{}' \in \Tab{}_t}\tab{}'$
\tcc*[f]{\hspace{-0.05em}For Abbreviations see Footnote~\ref{foot:crossproduct}.\hspace{-0.05em}}
	  
    \tcc*[h]{Get for a row~$\vec \tabval$ its combined child rows (origins)}\; %
    $\label{algo3:orig}\orig_t(\vec \tabval) \eqdef %
    \SB S %
    \SM S \in \text{CCR}_t, \vec u \in \tab{}, \tab{} =
    {\AlgW}(t,\chi(t),\prog_t,\atto,f_w(S)) %
    \SE$
      
    \tcc*[h]{Get for a table~$S$ of combined child rows its successors (evolution)
	 }\; $\label{algo3:evol}\evol_t(S) \eqdef %
    \SB \vec \tabval \SM %
    \vec \tabval \in \tab{}, %
    \tab{} = {\AlgC}(t,\chi(t),\prog_t,\atto, \tab{}'), \tab{}'\in
    S%
    \SE$\vspace{0.3em}
	
  \For{\label{algo3:for}\text{\normalfont iterate}
    $t$ in \text{\normalfont post-order}(T,n)}{\vspace{-0.05em}%
    %
	%
    \hspace{-0.5em}\tcc*[h]{Compute counter-witnesses ($\prec$-smaller rows) for a
      witness set~$M$ }\; $\label{algo3:subs}\hspace{-0.5em}\subsets_\prec(f, M, S) \eqdef %
    \SB %
    \vec \tabval \mid \vec \tabval \in \ATabs{\AlgC}{t},
	\vec \tabval \in
    \evol_t(f(S)), \vec \tabval = \langle C, \cdots \rangle, C\prec M \}$%
      %
    %
	
    \hspace{-0.5em}\tcc*[h]{Link each witness~$\vec \tabval$ to its counter-witnesses and store the results}\;
    $\label{algo3:link}\hspace{-0.5em}\ATabs{{\AlgW, \AlgC}}{t} \leftarrow %
    \SB ( \vec \tabval, %
    \subsets_{\subsetneq}(f_{w}, M, S) \cup \subsets_{\subseteq}(f_{cw}, M, S) ) \SM \vec \tabval\hspace{-0.1em} \in\hspace{-0.1em}
    \ATabs{\AlgW}{t}, 
    \vec \tabval\hspace{-0.05em} =\hspace{-0.05em} \langle M, \cdots \rangle, S\hspace{-0.05em} \in\hspace{-0.05em} \orig_t(\tabval)
    \SE$\; %
    \vspace{-0.5em} }\vspace{-0.15em}%
  \caption{Algorithm
    $\cw_{{\AlgW}, {\AlgC}}({\cal T}, \ATab{\AlgW}, \ATab{\AlgC})$ for
    linking counter-witnesses to witnesses.}
\label{fig:dpontd3}
\end{algorithm}


%
 %
%
\noindent Our algorithm (\MDP) executed as Step~3 runs $\dpa_{\INCSAT}$,
$\dpa_{\INCCOUNTERSAT}$ and $\cw_{\INCSAT, \INCCOUNTERSAT}$ in three
passes (3.I, 3.II, and 3.III) as follows:

\begin{enumerate}
\item[3.I.]\label{pass1} %
  First, we run the algorithm~$\dpa_{{\INCSAT}}$, which computes
  in a bottom-up traversal for every node~$t$ in the tree
  decomposition a table $\ATabs{\INCSAT}{t}$ of witnesses for~$t$.  Then, in a
  top-down traversal for every node~$t$ in the TD
  remove from tables $\ATabs{\INCSAT}{t}$ witnesses, which do not extend to a
  witness in the table for the parent node (``Purge non-witnesses'');
  these witnesses can never be used to construct a model (nor answer set) of the program.
\item[3.II.]\label{pass2} For this step, let ${\INCCOUNTERSAT}$ be a
  table algorithm computing only counter-witnesses of $\INC$ (blue parts of
  Algorithm~\ref{fig:incinc}).
  We execute $\dpa_{{\INCCOUNTERSAT}}$, compute for all witnesses
  counter-witnesses at once and store the resulting tables
  in~$\ATabs{\INCCOUNTERSAT}{\cdot}$. For every node~$t$, table~$\ATabs{\INCCOUNTERSAT}{t}$ contains
  counter-witnesses to witness being $\subset$-minimal. Again,
  irrelevant rows are removed (``Purge
  non-counter-witnesses'').

\item[3.III.]\label{pass3} Finally, in a bottom-up traversal for every
  node~$t$ in the TD, witnesses and counter-witnesses
  are linked using algorithm~$\cw_{{\INCSAT}, {\INCCOUNTERSAT}}$ (see
  Algorithm~\ref{fig:dpontd3}).  $\cw_{{\INCSAT}, {\INCCOUNTERSAT}}$ 
  takes previous results and maps rows in $\ATabs{\INCSAT}{t}$ to a table (set)
  of rows in $\ATabs{\INCCOUNTERSAT}{t}$.
  
\end{enumerate}%
\noindent We already explained %
the table algorithms~$\dpa_{\INCSAT}$ and $\dpa_{\INCCOUNTERSAT}$ in
the previous section. The main part of our multi-pass algorithm is the
algorithm~$\cw_{{\INCSAT}, {\INCCOUNTERSAT}}$ based on the general
algorithm~$\cw_{{\AlgW}, {\AlgC}}$ (Algorithm~\ref{fig:dpontd3}) with
$\AlgW=\INCSAT$, $\AlgC=\INCCOUNTERSAT$, which links those separate
tables together.  Before we quickly discuss the core
of~$\cw_{{\AlgW}, {\AlgC}}$ in
Lines~\ref{algo3:for}--\ref{algo3:link}, note that
Lines~\ref{algo3:ccr}--\ref{algo3:evol} introduce auxiliary
definitions.  Line~\ref{algo3:ccr} combines rows of the child nodes of
given node~$t$, which is achieved by a product over sets%
\footnote{\label{foot:crossproduct}%
  For set~$I=\{1,\ldots, n\}$ and sets~$S_i$, we define
  ${\prod _{i\in I}S_{i}\eqdef S_{1}\times \dotsm \times
    S_{n}=\{(s_{1},\dotsc ,s_{n}):s_{i}\in S_{i}\}}.$ Moreover, for
  ~${\prod_{i\in I}S_{i}}$, let
  ${\hat{\prod}_{i\in I}S_{i}} \eqdef \SB \Big\{\{s_1\},\dotsc,
  \{s_n\}\Big\} \SM (s_1, \dotsc, s_n) \in \prod_{i\in I}S_{i}
  \SE$. If for each~$S \in {\hat{\prod}_{i\in I}S_{i}}$ and
  $\{s_i\} \in S$, $s_i$ is a pair with a witness and a
  counter-witness part, %
  let $f_w(S)\eqdef \bigcup_{\{( W_i, C_i )\}\in S} \{\{W_i\}\}$
  and $f_{cw}(S)\eqdef \bigcup_{\{( W_i, C_i)\}\in S} \{\{C_i\}\}$
  restrict~$S$ to the resp.\ (counter-)witness parts.
}, %
where we drop the order and keep sets only. Line~\ref{algo3:orig}
concerns determining for a row~$\vec
\tabval$ its \emph{origins} (finding preceding combined rows that lead
to~$\vec \tabval$ using table
algorithm~$\AlgW$). Line~\ref{algo3:evol} covers deriving succeeding
rows for a certain child row combination its \emph{evolution} rows via
algorithm~$\AlgC$. In an implementation, origin as well as evolution
are not computed, but represented via pointer data structures directly
linking to $\ATabs{\AlgW}{\cdot}$ or
$\ATabs{\AlgC}{\cdot}$, respectively.  Then, the table
algorithm~$\cw_{{\AlgW},
  {\AlgC}}$ applies a post-order traversal and links witnesses to
counter-witnesses in Line~\ref{algo3:link}.  %
$\cw_{{\AlgW}, {\AlgC}}$ searches for
origins~($\orig$) of a certain witness~$\vec
\tabval$, uses the
counter-witnesses~($f_{cw}$) linked to these origins, and then
determines the
evolution~($\evol$) in order to derive counter-witnesses~(using
$\subsets$) of~$\vec \tabval$.%
\begin{figure}[t]%
  \centering %
\includegraphics[scale=1.158]{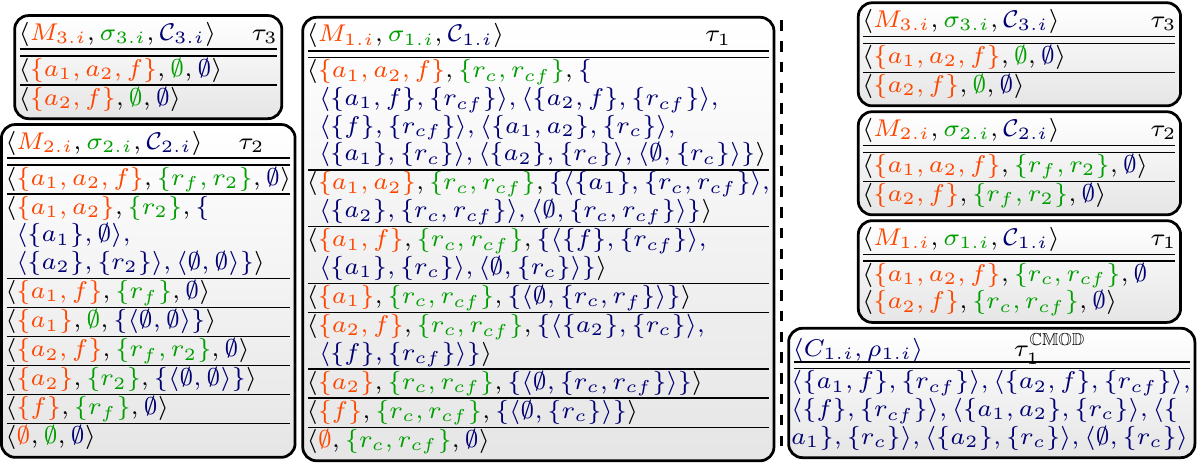}%
\caption{Selected DP tables after~$\dpa_{\INC}$~(left) and after~\MDP~(right) for
  TD~$\TTT$.}
\label{fig:exmdp}
\end{figure}%

\begin{theorem}
\label{thm:mdp}
For a program~$\prog$ of semi-incidence treewidth $k:=\tw{S(\prog)}$,
the algorithm~\MDP is correct and runs in
time~$\bigO{2^{2^{k+2}}\cdot \CCard{\prog}}$.
\end{theorem}%
\begin{proof}[Proof (Sketch)] %
  Due to space constraints, we only sketch the proof idea for
  enumerating answer sets of disjunctive ASP programs by means of
  \MDP. Let $\prog$ be a disjunctive program and
  $k\eqdef \tw{S(\prog)}$. We establish a reduction~$R(\prog,k)$ of
  \AspEnum to \MSatEnum, such that there is a one-to-one
  correspondence between answer sets and models of the formula, more
  precisely, for every answer set~$M$ of $\prog$ and for the resulting
  instance~$(F,k')=R(\prog, k)$ the set~$M \cup \{sol\}$ is a
  subset-minimal model of~$F$ and $k'=tw(I(F))$ with
  $k' \leq 7k + 2$. We compute in time~$2^{\bigO{k'^3}} \cdot \CCard{F}$ a TD of
  width at most~$k'$~\cite{BodlaenderKoster08} and add $sol$ to every bag.
  Using a table algorithm designed for SAT~\cite{SamerSzeider10b} %
  we compute witnesses and counter-witnesses. %
  Conceptually, one could also modify ${\INCSAT}$ for this
  task. %
  To finally show correctness of linking counter-witnesses to
  witnesses as presented in~$\cw_{\INCSAT, \INCSAT}$, we have to
  extend earlier work~\cite[Theorem~3.25 and 3.26]{BliemEtAl16b}.
  Therefore, we enumerate subset-minimal models of~$F$ by following
  each witness set containing $sol$ at the root having
  counter-witnesses~$\emptyset$ back to the leaves.  This runs in
  time~$\bigO{2^{2^{(7k+2)+2}}\cdot
    \CCard{\prog}}$,~c.f.,~\cite{BliemEtAl16b,FichteEtAl17}.  A more
  involved (direct) proof, allows to decrease the runtime
  to~$\bigO{2^{2^{k+2}}\cdot \CCard{\prog}}$ (even for choice %
  rules).  \hfill
\end{proof}

\begin{example}
  Let $k$ be some integer and $P_k$ be some program that contains the
  following rules $r_c\eqdef \{a_1, \cdots, a_k\} \hsep f$,
  $r_2\eqdef \hsep \neg a_2$, $\ldots$, $r_k \eqdef \hsep \neg a_{k}$,
  and $r_f\eqdef \hsep \neg f$ and~$r_{cf} \eqdef \{f\} \hsep$.  
  The rules~$r_1$, $\ldots$, $r_k$ simulate that only certain subsets
  of~$\{a_1, \cdots, a_k\}$ are allowed.  Rules~$r_f$ and~$r_{cf}$
  enforce that~$f$ is set to true.  Let ${\cal T} = (T,\chi,t_3)$ be a
  TD of the semi-incidence graph~$S({P}_k)$ of program~${P}_k$ where
  $T=(V,E)$ with $V=\{t_1,t_2,t_3\}$, $E=\{(t_1,t_2),(t_2,t_3)\}$,
  $\chi(t_1) = \{a_1, \cdots, a_k, f, r_c, r_{cf}\}$,
  $\chi(t_2) = \{a_1, \cdots, a_k, r_2, \cdots, r_k, r_f\}$, and
  $\chi(t_3) = \emptyset$.  Figure~\ref{fig:exmdp} (left) illustrates
  the tables for program~$P_2$ after $\dpa_{\INC}$, whereas
  Figure~\ref{fig:exmdp} (right) presents tables using $\MDP$, which
  are exponentially smaller in~$k$, mainly due to cleanup.  Observe
  that Pass 3.II~$\MDP$, ``temporarily'' materializes
  counter-witnesses \emph{only} for~$\tab{1}$, presented in
  table~$\tab{1}^\INCCOUNTERSAT$.  Hence, using multi-pass
  algorithm~$\MDP$ results in an exponential speedup.  Note that we
  can trivially extend the program such that we have the same effect
  for a TD of minimum width and even if we take the incidence graph.
  In practice, programs containing the rules above frequently
  occur when encoding by means of saturation~\cite{EiterGottlob95}.
  The program~$P_k$ and the TD~$\TTT$ also reveal that a different TD
  of the same width,
  where $f$ occurs already very early in the bottom-up traversal,
  would result in a smaller table~$\tab{1}$ even when
  running~$\dpa_{\INC}$. 
\end{example}

\subsection{Implementation Details}
Efficient implementations of dynamic programming algorithms on TDs are
not a by-product of computational complexity theory and involve tuning
and sophisticated algorithm engineering. %
Therefore, we present additional implementation details of
algorithm~$\MDP$ into our prototypical multi-pass solver
{\dynaspplus{\cdot}}, including two variations (depgraph, joinsize TDs).

Even though normalizing a TD can be achieved without increasing its
width, a normalization may artificially introduce additional
atoms. Resulting in several additional intermediate join nodes among
such artificially introduced atoms requiring a significant amount of
total unnecessary computation in practice.
On that account, we use \emph{non-normalized} tree decompositions.
In order to still obtain a fixed-parameter linear algorithm, we limit
the number of children per node to a constant.
Moreover, \emph{linking} counter-witnesses to witnesses efficiently is
crucial. The main challenge is to deal with situations where a row
(witness) might be linked to different set of counter-witnesses
depending on different predecessors of the row (hidden in set notation
of the last line in Algorithm~\ref{fig:dpontd3}). In these cases,
{\dynaspplus{\cdot}} eagerly creates a ``clone'' in form of a very
light-weighted proxy to the original row and ensures that only the
original row (if at all required) serves as counter-witness during
pass three. Together with efficient caches of counter-witnesses,
{\dynaspplus{\cdot}} reduces overhead due to clones in practice.
Dedicated \emph{data structures} are vital. Sets of Witnesses and
satisfied rules are represented in the {\dynaspplus{\cdot}} system via
constant-size bit vectors. $32$-bit integers are used to represent by
value~$1$ whether an atom is set to true or a rule is satisfied in the
respective bit positions according to the bag.
A restriction to $32$-bit integers seems reasonable as we assume
for now (practical memory limitations) that our
approach works well on TDs of width~$\leq 20$.
Since state-of-the-art computers handle such constant-sized integers
extremely efficient, {\dynaspplus{\cdot}} allows for efficient
projections and joins of rows, and subset checks in general.
In order to \emph{not recompute} counter-witnesses (in Pass~3.II) for
different witnesses, we use a three-valued notation of counter-witness
sets consisting of atoms set to true (T) or false (F) or false but
true in the witness set (TW) used to build the reduct.  Note that the
algorithm enforces that only (TW)-atoms are relevant, i.e., an atom
has to occur in a default negation or choice rule.

Minimum width is not the only optimization goal when computing TDs by
means of heuristics.  Instead, using TDs where a certain feature value
has been maximized in addition (\emph{customized TDs}) works seemingly
well in practice~\shortversion{\cite{AbseherEtAl17a}}\longversion{\cite{AbseherEtAl17a,MorakEtAl12}}.
While DynASP2.5 (\MDP) does not take additional TD features into
account, we also implemented a variant (\emph{\dynaspplus{\cdot}
  depgraph}), which prefers one out of ten TDs that intuitively
speaking avoids to introduce head atoms of some rule~$r$ in node~$t$,
without having encountered every body atom of~$r$ below~$t$, similar
to atom dependencies in the program~\cite{GottlobScarcelloSideri02}.
The variant \emph{\dynaspplus{\cdot} joinsize} minimizes bag sizes of
child nodes of join nodes,~c.f.~\cite{AbseherMusliuWoltran17a}.

\label{sec:evaluation}

\subsection{Experimental Evaluation}
\label{sec:experiments}
\noindent We performed experiments to investigate the runtime behavior of
{\dynaspplus{\cdot}} and its variants, in order to evaluate whether
our multi-pass approach can be beneficial and has practical advantages
over the classical single pass approach ({\dynasp{\cdot}}).
Further, we considered the dedicated ASP solver Clasp~3.3.0%
\longversion{\footnote{Clasp is available
  at~\url{https://github.com/potassco/clasp/releases/tag/v3.3.0}.}}.
Clearly, we \emph{cannot} hope to solve programs with graph
representations of high treewidth. However, programs involving
real-world graphs such as graph problems on transit graphs admit TDs
of acceptable width to perform DP on TDs.  
To get a first intuition, we focused on the Steiner tree problem
(\pname{St}) for our benchmarks.
Note that we support the most frequently used SModels input
format\longversion{~\cite{lparse}} for our implementation.

\newcommand{\blahtab}[1]{{\tiny{#1}}}
\begin{figure*}[t]
\includegraphics[scale=0.985]{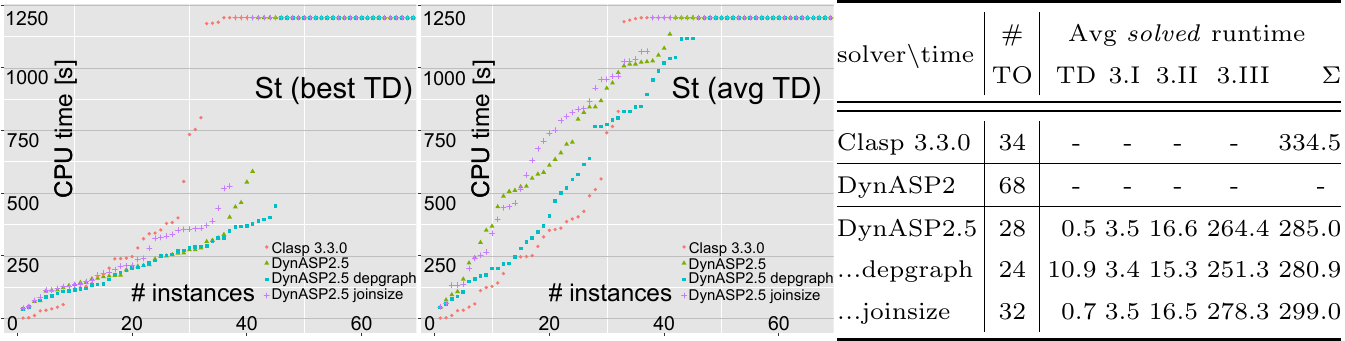}
\trash{
\begin{minipage}{0.3075\textwidth}
	\includegraphics[scale=0.27]{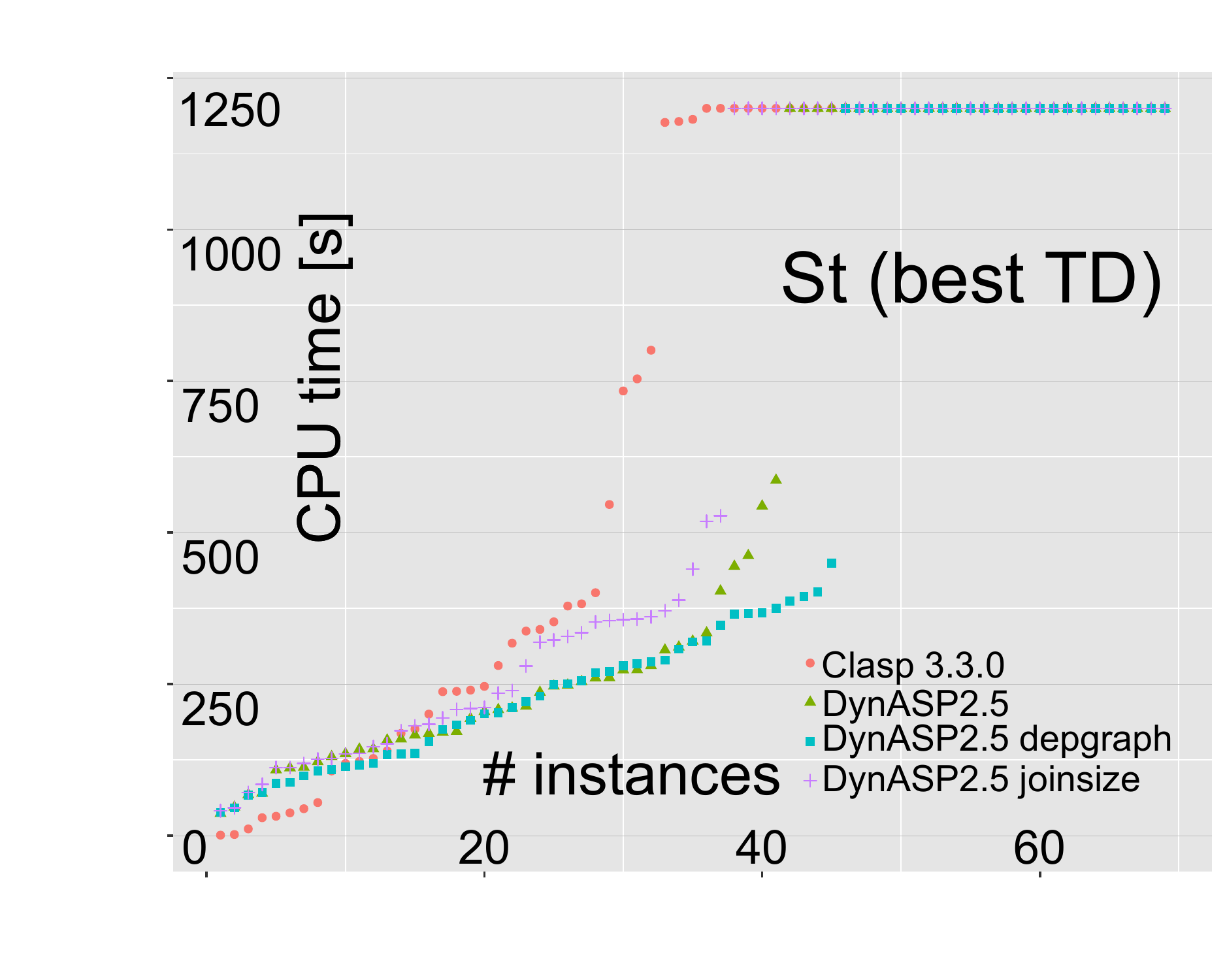}%
\end{minipage}
\begin{minipage}{0.31\textwidth}
	\includegraphics[scale=0.27]{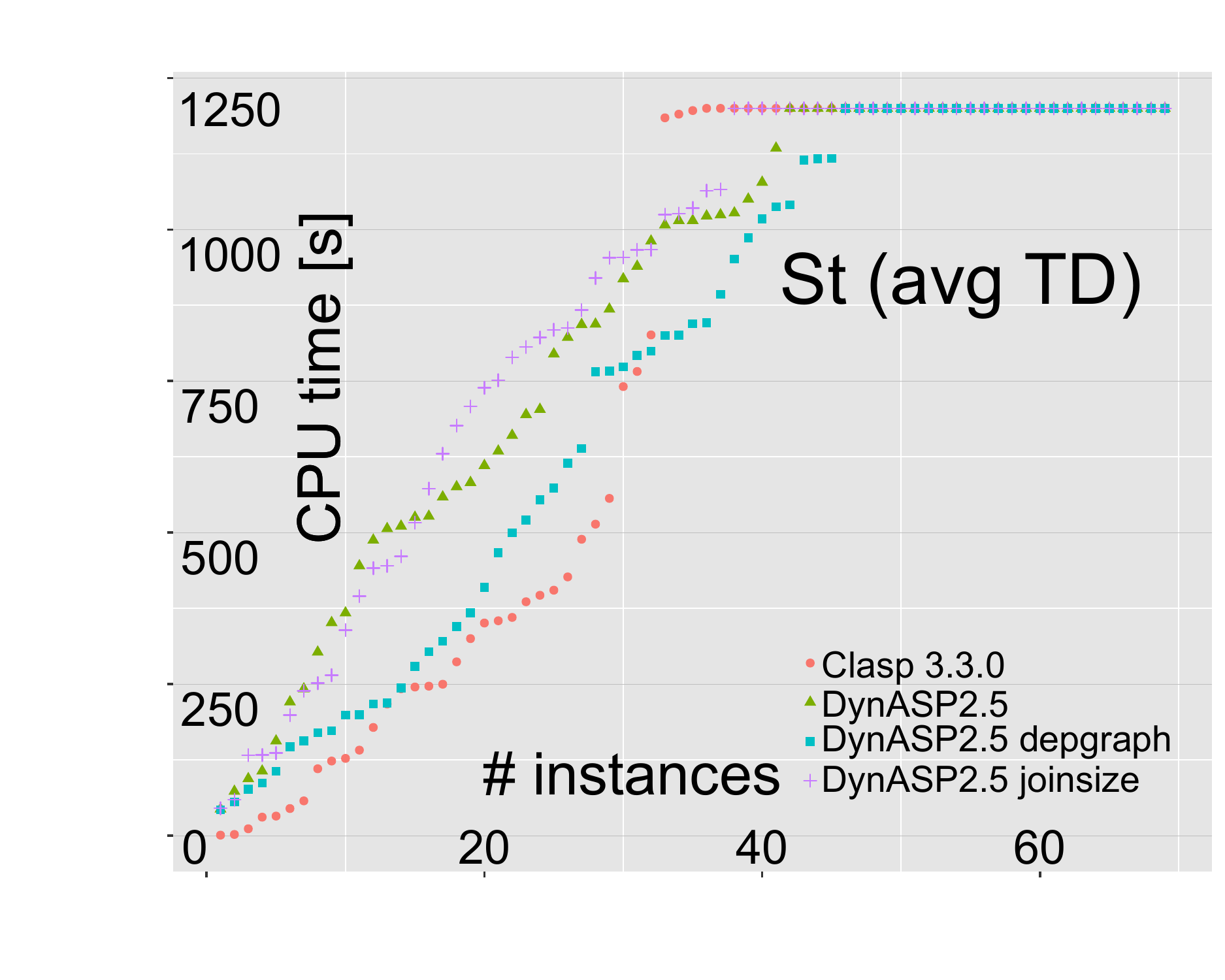}%
\end{minipage}
\begin{minipage}{0.37\textwidth}%
\begin{tabular}{l|c|r@{}
	r@{}%
	r@{}%
	r@{}r@{}}
    \toprule %
     \multirow{2}{*}{\blahtab{solver$\backslash$time}} & {\blahtab{$\#$}} &
				\multicolumn{5}{c}{\blahtab{Avg runtime  (\emph{solved} insts.)}}\\ & 
					{\blahtab{TO}} &
					\multicolumn{1}{c}{\blahtab{TD}} & 
					\multicolumn{1}{c}{\blahtab{3.I}} & 
					\multicolumn{1}{c}{\blahtab{3.II}} & 
					\multicolumn{1}{c}{\blahtab{3.III}} & 
					\multicolumn{1}{c}{\blahtab{$\Sigma$}}\\ 
    \midrule\midrule
  {\blahtab{Clasp~3.3.0}} & \blahtab{34} & \blahtab{-} & \blahtab{-} & \blahtab{-} & \blahtab{-} & \blahtab{334.5}\\\hline
  \blahtab{\dynaspplus{}} & \blahtab{28} & \blahtab{0.5} & \blahtab{3.5} & \blahtab{16.6} & \blahtab{264.4} & \blahtab{285.0}\\
  \\
  \blahtab{\dynaspplus{}} & \blahtab{28} & \blahtab{0.5} & \blahtab{3.5} & \blahtab{16.6} & \blahtab{264.4} & \blahtab{285.0}\\
  \blahtab{``depgraph''} & \blahtab{24} & \blahtab{10.9} & \blahtab{3.4} & \blahtab{15.3} & \blahtab{251.3} & \blahtab{280.9}\\
  \blahtab{``joinsize''} & \blahtab{32} & \blahtab{0.7} & \blahtab{3.5} & \blahtab{16.5} & \blahtab{278.3} & \blahtab{299.0}\\

    \bottomrule
  \end{tabular}%
\end{minipage}%
}
\caption{Cactus plots showing best and average runtime among five TDs
  (left). Number of Timeouts (TO) and average runtime among solved
  instances (right).}
\label{fig:random}
\end{figure*}%

We mainly inspected the CPU time using the average over five runs per
instance (five fixed seeds allow certain variance for heuristic TD
computation).
For each run, we limited the environment to 16 GB RAM and 1200 seconds
CPU time.
We used Clasp with options ``\text{-{-}}stats=2
\text{-{-}}opt-strategy=usc,pmres,disjoint,stratify \text{-{-}}opt-usc-shrink=min -q'', which
enable very recent improvements for unsatisfiable
cores~\cite{AlvianoDodaro16a}, and disabled
solution printing/recording. We also benchmarked Clasp with branch-and-bound,
which was, however, outperformed by the unsat.\
core options on all our instances.
Note that without the very recent unsatisfiable core advances Clasp
timed out on almost every instance.
We refer to an extended version~\cite{FichteEtAl17b} for more
details on the benchmark instances, encodings, and benchmark
environment. %
The left plot in Figure~\ref{fig:random} shows the result of always
selecting the best among five TDs, %
whereas the right plot concerns average runtime.  The table in
Figure~\ref{fig:random} reports on average running times (TD
computation and Passes~3.I, 3.II, 3.III) among the \emph{solved}
instances and the total number of timeouts (TO).  We consider an
instance to time out, when all five TDs exceeded the limit.
For the variants depgraph and joinsize, runtimes for computing and
selecting among ten TDs are included.
Our empirical benchmark results confirm that {\dynaspplus{\cdot}}
exhibits competitive runtime behavior even for TDs of treewidth
around~$14$.  Compared to state-of-the-art ASP solver Clasp,
{\dynaspplus{\cdot}} is capable of additionally delivering 
the number of optimal solutions.
In particular, variant ``depgraph'' shows promising runtimes.

\section{Conclusion}
\label{sec:conclusions}
In this paper, we presented a novel approach for ASP solving based on
ideas from parameterized complexity.  Our algorithms runs in linear
time assuming bounded treewidth of the input program. Our solver applies
DP in three passes, thereby avoiding redundancies. %
Experimental results indicate that our ASP solver is competitive for
certain classes of instances with small treewidth, where the latest
version of the well-known solver~Clasp hardly keeps up.  An
interesting question for future research is whether a linear
amount of passes (incremental DP) can improve the runtime
behavior. 
\bibliography{references}

\begin{thebibliography}{10}

\bibitem{AbseherEtAl17a}
M.~Abseher, F.~Dusberger, N.~Musliu, and S.~Woltran.
\newblock Improving the efficiency of dynamic programming on tree
  decompositions via machine learning.
\newblock {\em JAIR}, 58:829--858, 2017.

\bibitem{AbseherMusliuWoltran17a}
M.~Abseher, N.~Musliu, and S.~Woltran.
\newblock htd -- a free, open-source framework for (customized) tree
  decompositions and beyond.
\newblock In {\em CPAIOR'17}, 2017.
\newblock To appear.

\bibitem{AlvianoDodaro16a}
M.~Alviano and C.~Dodaro.
\newblock Anytime answer set optimization via unsatisfiable core shrinking.
\newblock {\em TPLP}, 16(5-6):533---551, 2016.

\bibitem{AlvianoEtAl13}
M.~Alviano, C.~Dodaro, W.~Faber, N.~Leone, and F.~Ricca.
\newblock {WASP}: A native {ASP} solver based on constraint learning.
\newblock In {\em LPNMR'13}, volume 8148 of {\em LNCS}, pages 54--66. Springer,
  2013.

\bibitem{BliemEtAl16b}
B.~Bliem, G.~Charwat, M.~Hecher, and S.~Woltran.
\newblock {D-FLAT\^{}2}: Subset minimization in dynamic programming on tree
  decompositions made easy.
\newblock {\em Fund. Inform.}, 147:27--34, 2016.

\bibitem{BodlaenderKoster08}
H.~Bodlaender and A.~M. C.~A. Koster.
\newblock Combinatorial optimization on graphs of bounded treewidth.
\newblock {\em The Computer J.}, 51(3):255--269, 2008.

\bibitem{Dell17a}
H.~Dell.
\newblock The 2st parameterized algorithms and computational experiments
  challenge -- {T}rack {A}: Treewidth.
\newblock Technical report, 2017.

\bibitem{DurandHK05}
A.~Durand, M.~Hermann, and P.~G. Kolaitis.
\newblock Subtractive reductions and complete problems for counting complexity
  classes.
\newblock {\em Th. Comput. Sc.}, 340(3), 2005.

\bibitem{EiterGottlob95}
T.~Eiter and G.~Gottlob.
\newblock On the computational cost of disjunctive logic programming:
  Propositional case.
\newblock {\em Ann. Math. Artif. Intell.}, 15(3--4):289--323, 1995.

\bibitem{FichteHecherMoraketal2016report_}
J.~K. Fichte, M.~Hecher, M.~Morak, and S.~Woltran.
\newblock {Answer Set Solving using Tree Decompositions and Dynamic Programming
  - The DynASP2 System -}.
\newblock Technical Report DBAI-TR-2016-101, {TU Wien}, 2016.

\bibitem{FichteEtAl17}
J.~K. Fichte, M.~Hecher, M.~Morak, and S.~Woltran.
\newblock Answer set solving with bounded treewidth revisited.
\newblock In {\em LPNMR'17}, 2017.
\newblock To appear.

\bibitem{FichteEtAl17b}
J.~K. Fichte, M.~Hecher, and S.~Woltran.
\newblock {DynASP2.5: Dynamic Programming on Tree Decompositions in Action}.
\newblock {\em CoRR}, abs/cs/arXiv:1702.02890, 2017.

\bibitem{BomansonGebserJanhunen16}
M.~Gebser, J.~Bomanson, and T.~Janhunen.
\newblock Rewriting optimization statements in answer-set programs.
\newblock {\em Technical Communications of ICLP 2016}, 2016.

\bibitem{GebserMarateaRicca16a}
M.~Gebser, M.~Maratea, and F.~Ricca.
\newblock What's hot in the answer set programming competition.
\newblock In {\em AAAI'16}, pages 4327--4329. The AAAI Press, 2016.

\bibitem{GottlobScarcelloSideri02}
Georg Gottlob, Francesco Scarcello, and Martha Sideri.
\newblock Fixed-parameter complexity in {AI} and nonmonotonic reasoning.
\newblock {\em AIJ}, 138(1-2):55--86, 2002.

\bibitem{HermannP09}
M.~Hermann and R.~Pichler.
\newblock Complexity of counting the optimal solutions.
\newblock {\em Th. Comput. Sc.}, 410(38--40), 2009.
\newblock URL: \url{http://dx.doi.org/10.1016/j.tcs.2009.05.025}.

\bibitem{JaklPichlerWoltran09}
M.~Jakl, R.~Pichler, and S.~Woltran.
\newblock Answer-set programming with bounded treewidth.
\newblock In {\em IJCAI'09}, 2009.

\bibitem{JanhunenNiemela16a}
T.~Janhunen and I.~Niemel{\"a}.
\newblock The answer set programming paradigm.
\newblock 2016.

\bibitem{clasp}
B.~Kaufmann, M.~Gebser, R.~Kaminski, and T.~Schaub.
\newblock clasp -- a conflict-driven nogood learning answer set solver, 2015.

\bibitem{KochL99}
C.~Koch and N.~Leone.
\newblock Stable model checking made easy.
\newblock In {\em IJCAI'99}, 1999.

\bibitem{MorakEtAl12}
M.~Morak, N.~Musliu, R.~Pichler, S.~R{\"u}mmele, and S.~Woltran.
\newblock Evaluating tree-decomposition based algorithms for answer set
  programming.
\newblock In {\em LION'12}, LNCS, pages 130--144. Springer, 2012.

\bibitem{SamerSzeider10b}
M.~Samer and S.~Szeider.
\newblock Algorithms for propositional model counting.
\newblock {\em J. Desc. Alg.}, 8(1), 2010.

\bibitem{SimonsNiemelaSoininen02}
P.~Simons, I.~Niemel{\"a}, and T.~Soininen.
\newblock Extending and implementing the stable model semantics.
\newblock {\em AIJ}, 138(1-2):181--234, 2002.

\bibitem{lparse}
T.~Syrj{\"a}nen.
\newblock Lparse 1.0 user's manual.
\newblock \url{tcs.hut.fi/Software/smodels/lparse.ps}, 2002.

\end{thebibliography}
\longversion{ 
\clearpage
\appendix
\section{Additional Example}%

\noindent We assume again row numbers per table~$\tab{t}$, i.e.,
$\tabval_{t.i} = \langle M_{t.i}, \sigma_{t.i}, {\cal C}_{t.i}
\rangle$ is the $i^{th}$-row. 
Further, for each
counter-witness~
$\langle {C_{t.i.j}, \rho_{t.i.j} \rangle \in {\cal C}_{t.i}}$, $j$
marks its ``order'' (as depicted in Figure~\ref{fig:running1_inc}
(right)) in set~${\cal C}_{t.i}$.

\begin{example}\label{ex:running_counter}
  Again, we consider~$\prog$ of Example~\ref{ex:running} and
  ${\cal T}=(\cdot, \chi)$ of Figure~\ref{fig:running1_inc} as well as
  tables~$\tab{1}$, $\ldots$, $\tab{34}$ of
  Figure~\ref{fig:running1_inc}~(right) using $\dpa_\INC$.  We only
  discuss certain tables.
  Table~$\tab{1}=\{\langle \emptyset, \emptyfunc, \emptyset \rangle\}$
  as $\type(t_1)=\leaf$. Node~$t_2$ introduces atom~$e_{ab}$,
  resulting in table
  $\{\langle \{e_{ab}\}, \emptyset, \{(\emptyset, \emptyfunc)\}
  \rangle, \langle \emptyset, \emptyset, \emptyset \rangle \}$
  (compare to Algorithm~\ref{fig:incinc} L3--5).  Then, node~$t_3$
  introduces rule~$r_{ab}$, which is removed in node~$t_4$.  Note
  that~$C_{3.1.1} = \langle \emptyset, \emptyset \rangle \in {\cal
    C}_{3.1.1}$ does not have a ``successor row'' in table~$\tab{4}$
  since $r_{ab}$ is not satisfied (see L11 and
  Definition~\ref{def:bagreduct}).
  Table~$\tab{6}$ is then the result of a chain of introduce nodes,
  and contains for each witness set~$M_{6.i}$ every possible
  counter-witness set~$C_{6.i.j}$ with~$C_{6.i.j}\subsetneq M_{6.i}$.
  We now discuss table~$\tab{12}$, intuitively containing (a
  projection of) (counter-)witnesses of~$\tab{10}$, which satisfy
  rule~$r_{bc}$ after introducing rule~$r_{c1}$.  Observe that there
  is no succeeding witness set for~$M_{6.2}=\{e_{ab}\}$ in~$\tab{10}$
  (nor~$\tab{12}$), since $e_{ab}\in M_{6.2}$, but
  $a_b\not\in M_{6.2}$ (required to satisfy~$r_b$).
  Rows~$\tabval_{12.1}, \tabval_{12.4}$ form successors of
  $\tabval_{6.3}$, while rows~$\tabval_{12.2}, \tabval_{12.5}$ succeed
  $\tabval_{6.1}$, since counter-witness set~$C_{6.1.1}$ has no
  succeeding row in~$\tab{10}$ because it does not satisfy~$r_b$.
  Remaining rows~$\tabval_{12.3}$, $\tabval_{12.6}$
  have ``origin''~$\tabval_{6.4}$ in~$\tab{6}$.
\end{example}

\section{Omitted Proofs}

\subsection{Correctness of {\dynaspplus{}}}
Bliem~\etal~\cite{BliemEtAl16b} have shown that augmentable
$\ATab{{\AlgW}}$ can be transformed into $\ATab{{\AlgW,\AlgW}}$, which
easily allows reading off subset-minimal solutions starting at the
table $\ATabs{{\AlgW, \AlgW}}{n}$ for TD root~$n$.
We follow their concepts and define a slightly extended variant of
\emph{augmentable} tables.
Therefore, we reduce the problem of enumerating disjunctive programs
to $\MSatEnum$ and show that the resulting tables of algorithm~\INCSAT
(see Algorithm~\ref{fig:incinc}) are \emph{augmentable}. In the end,
we apply an earlier theorem~\cite{BliemEtAl16b} transforming
$\ATab{\INCSAT}$ obtained by $\dpa_\INCSAT$ into
$\ATab{\INCSAT,\INCSAT}$ via the \emph{augmenting} function
$aug(\cdot)$ proposed in their work.  To this extent, we use auxiliary
definitions~$\Tab{}_t$, $\orig_t(\cdot)$ and $\evol_t(\cdot)$
specified in Algorithm~\ref{fig:dpontd3}.

\begin{definition}%
\label{def:computation}
Let~$\TTT = (T, \chi)$ be a TD where $T = (N,\cdot,\cdot)$, ${\AlgW}$ be a
table algorithm, $t \in N$, and $\tab{} \in \ATab{\AlgW}[t]$ be the
table for node~$t$. %
For tuple~$\vec \tabval = \langle M, \sigma, \cdots \rangle\in\tab{}$,
we define~$\alpha(\vec \tabval) \eqdef M$,
$\beta(\vec \tabval) \eqdef \sigma$.  We inductively
define~\[\alpha^*(\tab{}) \eqdef \bigcup_{\vec u \in \tab{}}
  \alpha(\vec u) \cup \bigcup_{\tab{}'\in\Tab{}_t} \alpha^*(\tab{}'),
  \text{ and }\]
\[\beta^*(\tab{}) \eqdef \bigcup_{\vec u \in
    \tab{}} \beta(\vec u) \cup \bigcup_{\tab{}'\in\Tab{}_t}
  \beta^*(\tab{}').\]
Moreover, we inductively define the \emph{extensions} of a row
$\vec \tabval \in \tab{}$ as
\[\E(\vec \tabval) \eqdef \bigg\SB \{\vec \tabval\} \cup U \mid U \in
  \bigcup_{\{\{\vec \tabval'_1\},\dots,\{\vec \tabval'_k\}\} \in
    \orig_t(\vec \tabval)} \{\tab{1} \cup \dots \cup \tab{k} \mid
  \tab{i} \in \E(\vec \tabval'_i) \text{ for all } 1 \leq i \leq
  k\}\bigg\SE.\]
\end{definition}

\begin{remark}
  Any extension~$U \in \E(\vec \tabval)$ contains $\vec \tabval$ and
  exactly one row from each table that is a descendant of $\tab{}$.
  If $\vec \tabval$ is a row of a leaf table,
  $\E(\vec \tabval) = \{\{\vec \tabval\}\}$ since
  $\orig_t(\vec \tabval) = \{\emptyset\}$
  assuming~$\prod_{i\in\emptyset}S_i = \{()\}$.
\end{remark}

\begin{definition}%
\label{def:solution}
Let $\tab{n}$ be the table in $\ATab{{\AlgW}}$ for TD root~$n$.  We
define the set~$\sol($\ATab{{\AlgW}}$)$ of \emph{solutions} of
$\ATab{{\AlgW}}$ as
$\sol($\ATab{{\AlgW}}$) \eqdef \{\alpha^*(U) \mid \vec \tabval \in
\tab{n},\; U \in \E(\vec \tabval)\}$
\end{definition}

\begin{definition}%
  Let $\tab{}$ be a table in $\ATab{{\AlgW}}$ such that
  $\tab{1}', \dots, \tab{k}'$ are the child tables~$\Tab{}_t$ and let
  $\vec \tabval, \vec v \in \tab{}$.  We say that
  $x \in X(\vec \tabval)$ has been \emph{$X-$illegally introduced} at
  $\vec \tabval$ if there are
  $\{\{\vec \tabval'_1\},\dots,\{\vec \tabval'_k\}\} \in \orig_t(\vec
  \tabval)$ such that for some $1 \leq i \leq k$ it holds that
  $x \notin X(\vec \tabval'_i)$ while $x \in X^*(\tab{i}')$.
  Moreover, we say that $x \in X(\vec v) \setminus X(\vec \tabval)$ has
  been \emph{$X-$illegally removed} at $\vec \tabval$ if there is some
  $U \in \E(\vec \tabval)$ such that $x \in X(U)$.
\end{definition}

\begin{definition}%
\label{def:flat}
We call a table $\tab{}$ \emph{augmentable} if the following conditions hold:
\begin{enumerate}
\item For all rows of the form $\langle M, \cdots, {\cal C} \rangle$,
  we have ${\cal C} = \emptyset$.
\item For all $\vec \tabval, \vec v \in \tab{}$ with
  $\vec \tabval \neq \vec v$ it holds that
  $\alpha(\vec \tabval) \cup \beta(\vec \tabval) \neq \alpha(\vec v)
  \cup \beta(v)$.
\item For all
  $\vec \tabval = \langle M, \sigma, \cdots \rangle \in \tab{}$,
  $\{\{\vec \tabval'_1\},\dots,\{\vec \tabval'_k\}\} \in \orig_t(\vec
  \tabval)$, $1 \leq i < j \leq k$, $I \in \E(\vec \tabval'_i)$ and
  $J \in \E(\vec \tabval'_j)$ it holds that
  $\alpha^*(I) \cap \alpha^*(J) \subseteq M$ and
  $\beta^*(I) \cap \beta^*(J) \subseteq \sigma$.
\item No element of $\alpha^*(\tab{})$ has been $\alpha$\hy illegally
  introduced and no element of $\beta^*(\tab{})$ has been $\beta$\hy
  illegally introduced.
\item No element of $\alpha^*(\tab{})$ has been $\alpha$\hy illegally
  removed and no element of $\beta^*(\tab{})$ has been $\beta$\hy
  illegally removed.
\end{enumerate}
We call $\ATab{{\AlgW}}$ \emph{augmentable} if all its tables are augmentable.
\end{definition}

\trash{
\noindent We now define a transformation of the result obtained by Algorithm~\ref{fig:incinc}
to computations and show that in fact the resulting computation is augmentable.

\begin{definition}
We denote by $trans({\AlgW}, \tab{t}, \chi(t), \prog_t, \atto, \Tab{})$ the transformation of table~$\tab{t}$ using algorithm~${\AlgW}$ to tables (in the sense of Definition~\ref{def:computation}) as follows.
Let $trans(\tab{t})$ be a table where for each element of $\langle M, \dots \rangle \in \tab{t}$, we set~$D(r)\eqdef M, S(r) \eqdef \emptyset$
and $P(r) \eqdef \{ p \mid \tab{t_i}\in \Tab{}, t_1 < \cdots < t_k, p \in \tab{t_1} \times \cdots \times \tab{t_k}, r = {\AlgW}(t, \chi(t), \prog_t, \atto, \{p\})\}$.
\end{definition}

\begin{definition}
We denote by $trans({\AlgW})$ the resulting computation obtained by following the scheme of Algorithm~\ref{fig:dpontd} 
and constructing $trans({\AlgW}, \tab{t}, \chi(t), \prog_t, \atto, \Tab{})$ in every iteration.
\end{definition}
\begin{observation}
Transformation $trans(\INCSAT)$ is an augmentable computation.
\end{observation}}

\noindent It is easy to see that $\ATab{\INCSAT}$ are augmentable, that is,
Algorithm~\ref{fig:incinc} ($\dpa_{\INCSAT}(\cdot)$) computes only
augmentable tables.
\begin{observation}
$\ATab{\INCSAT}$ are augmentable, since $\dpa_{\INCSAT}(\cdot)$ computes
augmentable tables.
$\ATab{\INCCOUNTERSAT}$ are augmentable, since
$\dpa_{\INCCOUNTERSAT}(\cdot)$ computes augmentable tables.
\end{observation}

The following theorem establishes that we can reduce an instance of
\AspEnum (restricted to disjunctive input programs) when parameterized
by semi-incidence treewidth to an instance of \MSatEnum when
parameterized by the treewidth of its incidence graph.

\newcommand{\Cc}{F}
\newcommand{\Ccp}{F^*}
\newcommand{\F}[1]{F^\text{\normalfont #1}}
\begin{lemma}
\label{thm:asptomm}
Given a disjunctive program~$\prog$ of semi-incidence
treewidth~$k = \tw{S(\prog)}$. We can produce in
time~$\bigO{\CCard{\prog}}$ a propositional formula~$F$ such that the
treewidth~$k'$ of the incidence graph~$I(F)$\footnote{The
  \emph{incidence graph}~$I(F)$ of a propositional formula~$F$ in CNF
  is the bipartite graph that has the variables and clauses of~$F$ as
  vertices and an edge~$v\, c$ if $v$ is a variable that occurs in $c$
  for some clause~$c \in F$~\cite{SamerSzeider10b}.}
is $k' \leq 7k+2$ and the answer sets of~$\prog$ and subset-minimal
models of~$\F{*}$ are in a particular one-to-one correspondence. More
precisely, $M$ is an answer set of $\prog$ if and only if
$M \cup M_{\text{aux}} \cup \{sol\}$ is a subset-minimal model of~$F$
where $M_{\text{aux}}$ is a set of additional variables occurring in~$F$, 
but not in~$P$ and 
variables introduced by Tseitin
normalization. %
\end{lemma}
\begin{proof}
  Let $P$ be a disjunctive program of semi-incidence
  treewidth~$k = \tw{S(\prog)}$. First, we construct a formula~$F$
  consisting of a conjunction over formulas~$\F{r}$, $\F{impl}$,
  $\F{sol}$, $\F{min}$ followed by Tseitin normalization of~$F$ to
  obtain~$\F{*}$.
  Among the atoms\footnote{Note that we do not distinguish between
    atoms and propositional variables in terminology here.} of our
  formulas will the atoms~$\at(\prog)$ of the program. Further, for
  each atom~$a$ such that $a \in B^-(r)$ for some rule~$r \in \prog$,
  we introduce a fresh atom~$a'$.  In the following, we denote by $Z'$
  the set $\{z': z\in Z\}$ for any set $Z$ and by
  $B^-_\prog \eqdef \bigcup_{r\in\prog} B^-_r$.  Hence, $(B^-_\prog)'$
  denotes a set of fresh atoms for atoms occurring in any negative
  body.  Then, we construct the following formulas:
  \begin{align}
    \F{r}(r) \eqdef& H_r \vee \neg B^+_r \vee (B^-_r)' && \text{for\ } r\in\prog\\
    \F{impl}(a) \eqdef & a \rightarrow a' && \text{for\ } a\in B^-_\prog\\
    \F{sol}(a) \eqdef & sol \rightarrow (a' \rightarrow a) &&\text{for\ } a\in B^-_\prog\\
    \F{min} \eqdef & \neg sol \rightarrow \bigvee_{a' \in (B^-_\prog)'} (a' \wedge \neg a)\\
    \F{} \eqdef & \bigwedge_{r\in\prog}\F{r}(r) \wedge \bigwedge_{a\in B^-_\prog}\F{impl}(a) \wedge \bigwedge_{a\in B^-_\prog}\F{sol}(a) \wedge \F{min}
  \end{align}

  Next, we show that $M$ is an answer set of $\prog$ if and only if
  $M \cup ([M \cap B^-_\prog])' \cup \{sol\}$ is a subset-minimal
  model of $F$.

  \ponlyif: Let $M$ be an answer set of~$\prog$. %
  We transform~$M$ into
  $Y \eqdef M \cup ([M \cap B^-_\prog])' \cup \{sol\}$.  Observe
  that~$Y$ satisfies all subformulas of~$F$ and therefore
  $Y \models F$.  It remains to show that~$Y$ is a minimal model
  of~$F$.  Assume towards a contradiction that~$Y$ is not a minimal
  model.  Hence, there exists~$X$ with $X \subsetneq Y$. We
  distinguish the following cases:
  \begin{enumerate}
  \item $sol\in X$: By construction of $F$ we have
    $X \models a' \leftrightarrow a$ for any~$a'\in (B^-_\prog)'$, which implies
    that %
    $X \cap \at(\prog) \models \prog^{M}$. However, this contradicts
    our assumption that~$M$ is an answer set of~$\prog$.
  \item $sol\not\in X$: By construction of $F$ %
    there is at least one atom~$a\in B^-_\prog$ with $a'\in X$, but $a\not\in
    X$. Consequently, $X \cap \at(\prog) \models \prog^{M}$. This
    contradicts again that~$M$ is an answer set of~$\prog$.
\end{enumerate}

\pif: %
Given a formula~$F$ that has been constructed from a program~$\prog$
as given above. Then, let $Y$ be a subset-minimal model of~$F$ such
that $sol \in Y$. By construction we have for
every~$a' \in Y \cap (B^-_\prog)'$ that
$a \in Y$. Hence, we let
$M = \at(\prog) \cap Y$. %
Observe that~$M$ satisfies every rule~$r\in\prog$ according to (A1)
and is in consequence a model of $\prog$.  It remains to show that~$M$
is indeed an answer set. Assume towards a contradiction that~$M$ is
not an answer set.  Then there exists a model~$N\subsetneq M$ of the
reduct~$\prog^M$.  We distinguish the following cases:
\begin{enumerate}
\item $N$ is not a model of~$\prog$: We construct
  $X \eqdef N \cup [Y \cap (B^-_\prog)']$ and show that~$X$ is indeed
  a model of~$F$.  For this, for every $r \in \prog$ where
  $B^-(r) \cap M \neq \emptyset$ we have $X\models \F{r}(r)$, since
  $(Y \cap (B^-_\prog)') \subseteq X$ by definition of~$X$.  For
  formulas~$(A1)$ constructed by $\F{r}(r)$ using remaining rules~$r$, we
  also have $X\models \F{r}(r)$, since $N\models\{r\}^M$. In conclusion,
  $X \models F$ and $X\subsetneq Y$, and therefore $X$ contradicts~$Y$
  is a subset-minimal model of~$F$.
\item $N$ is also a model of~$\prog$: Observe that then
  $X \eqdef N \cup [N \cap B^-_\prog]' \cup \{sol\}$ is also a model
  of~$F$, which contradicts optimality of~$Y$ since $X \subsetneq Y$.
\end{enumerate}

\noindent By Tseitin normalization, we obtain~$\Ccp$, thereby
introducing fresh atoms $l_{a'}$ for each $a'\in (B^-_\prog)'$:

\begin{align}
  \F{r*}(r) \eqdef & H_r \vee \neg B^+_r \vee (B^-_r)' && \text{for\ } r\in \prog \tag*{(A1*)}\\
  \F{impl*}(a)  \eqdef & \neg a \vee a' && \text{for\ } a\in B^-_\prog \tag*{(A2*)}\\
  \F{sol*}(a)  \eqdef & \neg sol \vee (\neg a' \vee a) && \text{for\ } a\in B^-_\prog \tag*{(A3*)}\\
  \F{min}_1  \eqdef & sol \vee \bigvee_{a' \in B^-_{\prog'}} (l_{a'}) \tag*{(A4.1*)}\\
  \F{min}_2(a)  \eqdef &\neg a' \vee a \vee l_{a'} && \text{for\ } a \in B^-_\prog \tag*{(A4.2*)}\\
  \F{min}_3(a) \eqdef &\neg l_{a'} \vee a' && \text{for\ } a \in B^-_\prog \tag*{(A4.3*)}\\
  \F{min}_4(a) \eqdef &\neg l_{a'} \vee \neg a && \text{for\ } a \in B^-_\prog \tag*{(A4.4*)}
\end{align}

\noindent Observe that the Tseitin normalization is correct and that
there is a bijection between models of $\Ccp$ and $\Cc$.  %

Observe that our transformations runs in linear time and that the size of~$\F{*}$
is linear in $\CCard{\prog}$. %
It remains to argue that $tw(I(\F{*})) \leq 7k+2$.  For this, assume
that ${\cal T}=(T,\chi,n)$ is an arbitrary but fixed TD of $S(\prog)$
of width~$w$.
We construct a new TD~${\cal T'} \eqdef (T,\chi',n)$ where $\chi'$ is
defined as follows.  For each TD node~$t$, %
\[\chi'(t)\eqdef \bigcup_{a\in B^-_\prog\cap \chi(t)} \{a', l_{a'}\} \cup [\chi(t) \cap \at(\prog)] \cup \{sol\} \cup cl(t)\]
where 
\[cl(t) \eqdef \bigcup_{a\in B^-_\prog\cap \chi(t)} [\F{impl*}(a),
  \F{sol*}(a), \F{min}_2(a), \F{min}_3(a), \F{min}_4(a)] \cup
  \{\F{min}_1\} \cup \bigcup_{r\in \prog\cap \chi(t)} \{\F{r*}(r)\}.\]
It is easy to see that~${\cal T'}$ is indeed a TD for $I(\Ccp)$ and
that the width of~${\cal T'}$ is at most~$7w + 2$.
\hfill
\end{proof}

\trash{
\begin{corollary}
Enumerating answer sets of any disjunctive ASP program~$\prog$ (problem \AspEnum) of bounded primal treewidth~$k$ ($tw(P(\prog)) = k$)
can be reduced to enumerating \MSat\ models (problem \MSatEnum) in time~$f(k)\cdot \CCard{\prog}$.
\end{corollary}}

\trash{
\begin{proposition}
\MSatEnum is \SIGMA{2}{p} complete, \MSatEnum\ is $F\SIGMA{2}{p}$ complete.
\end{proposition}
\begin{proof}
Membership of \MSatEnum can be shown via stating a non-deterministic algorithm with an $\NP$ oracle.
Completeness follows from Theorem~\ref{thm:asptomm}.
\end{proof}}

\begin{definition}%
\label{def:augc}
We inductively define an \emph{augmenting} function $\aug(\ATab{{\AlgW}})$ that maps each table $\tab{}\in\ATab{{\AlgW}}[t]$ for node $t$ 
from an augmentable table to a table in $\ATab{{\AlgW, \AlgW}}[t]$.
Let %
the child tables of $\tab{}$ be called $\tab{1}', \dots, \tab{k}'$.
For any $1 \leq i \leq k$ and $\vec \tabval \in \tab{i}$, we write $\res(\vec \tabval)$ to denote
$\{\vec v \in \aug(\tab{i}') \mid \alpha(\vec \tabval) = \alpha(\vec v)\}$.
We define $\aug(\tab{})$ as the smallest table that satisfies the following conditions:
\begin{enumerate}
\item
For any $\vec \tabval \in \tab{}$, $\{\{\vec \tabval'_1\}, \dots, \{\vec \tabval'_k\}\} \in \orig_t(\vec \tabval)$ and 
$\{\{\vec v'_1\}, \dots, \{\vec v'_k\}\} \in {\hat{\prod}_{1 \leq i \leq k}} \res(\vec \tabval'_i)$,
there is a row $\vec v \in \aug(\tab{})$
with
$\alpha(\vec \tabval) = \alpha(\vec v)$ and
$\{\{\vec v'_1\}, \dots, \{\vec v'_k\}\} \in \orig_t(\vec v)$.
\item For any $\vec \tabval, \vec v \in \aug(\tab{})$ with
  $\vec \tabval=\langle \cdots, {\cal C} \rangle$ such that
  $\alpha(\vec v) \subseteq \alpha(\vec \tabval)$,
  $\{\{\vec \tabval'_1\}, \dots, \{\vec \tabval'_k\}\} \in
  \orig_t(\vec \tabval)$ and
  $\{\{\vec v'_1\}, \dots, \{\vec v'_k\}\} \in \orig_t(\vec v)$ the
  following holds: Let $1 \leq i \leq k$ with
  $\vec \tabval'_i=\langle \cdots, {\cal C}_i \rangle$,
  $\vec c_i=\langle C_i, \cdots \rangle \in ({\cal C}_i \cup \{\vec
  \tabval'_i\})$ with $C_i \subseteq \alpha(\vec v'_i)$,
  and $1 \leq j \leq k$ with $\vec c_j \neq \vec
  \tabval'_j$ or $\alpha(\vec v) \subsetneq \alpha (\vec u)$.
  Then, there is a row $\vec c \in {\cal C}$ with
  $\alpha(\vec c) \subseteq \alpha(\vec v)$ if and only if $\vec
  c\in\tab{}$ and $\{\{\vec c_1\}, \dots, \{\vec c_k\}\} \in
  \orig_t(\vec c)$.
\end{enumerate}
For $\ATab{{\AlgW}}$, we write $\aug(\ATab{{\AlgW}})$ to denote 
the result isomorphic to $\ATab{{\AlgW, \AlgW}}$ 
where each table $\tab{}$ in $\ATab{{\AlgW}}$ corresponds to $\aug(\tab{})$.
\end{definition}

\begin{proposition}%
\label{thm:aug-solutions-minimal}
Let $\ATab{{\AlgW}}$ be augmentable.  Then,
\[\sol(\aug(\ATab{{\AlgW}})) = \{M \in \sol(\ATab{{\AlgW}}) \mid
  \nexists M' \in \sol(\ATab{{\AlgW}}): M' \subsetneq M\}.\]
\end{proposition}
\begin{proof}[Proof (Sketch)]
  The proof follows previous work~\cite{BliemEtAl16b}. We sketch only
  differences from their work. Any row~$\vec
  \tabval\in\tab{}$ of any
  table~$\tab{}$ not only consists of set~$\alpha(\vec
  \tabval)$ being subject to subset-minimization and relevant to
  solving \AspEnum.  In addition, our definitions presented above also
  allow ``\emph{auxiliary}'' sets~$\beta(\vec \tabval)$ per row~$\vec
  \tabval$, which are not subject to the minimization.  Moreover, by
  the correctness of the table algorithm \IINC by Fichte and
  Szeider~\cite{FichteEtAl17}, we only require to store a set~${\cal
    C}$ of counter-witnesses~$\langle C, \cdots \rangle \in {\cal
    C}$ per witness set~$M$, where each
  $C$ forms a \emph{strictly} $\subset$-smaller model
  of~$M$. As a consequence, there is no need to differ between sets of
  counter-witnesses, which are strictly included or not,
  see~\cite{BliemEtAl16b}.  Finally, we do not need to care about
  duplicate rows (solved via compression function
  $compr(\cdot)$ in~\cite{BliemEtAl16b}) in $\tab{}$, since
  $\tab{}$ is a set.  \hfill
\end{proof}

\begin{theorem}\label{thm:disjaug}
  \AspEnum when the input is restricted to disjunctive programs can be
  solved in time~$2^{2^{(7k+4)}} \cdot \CCard{\prog}$ computing
  $\aug(\dpa_\INCSAT(\cdot))$, where~$k$ refers to the treewidth of
  $S(\prog)$. %
\end{theorem}

\begin{proof}[Proof (Sketch)]
  First, we use reduction~$R(\prog,k)=(\Ccp, k')$ defined in
  Lemma~\ref{thm:asptomm} to construct an instance of SAT given our
  disjunctive ASP program~$\prog$.  Note
  that~$k'=tw(I(\Ccp))\leq 7k + 2$.  Then, we can compute in
  time~$2^{\bigO{k'^3}} \cdot \Card{I(\Ccp)}$ a tree decomposition of
  width at most~$k'$~\cite{BodlaenderKoster08}. Note that since we require
  to look for solutions containing~$sol$ at the root, we modify each
  bag of~${\cal T}$ such that it contains~$sol$. We call the resulting
  tree decomposition~${\cal T'}$.  We compute
  $\aug(\dpa_\INCSAT({\cal T'}))$ using formula~$\Ccp$ (see
  Algorithm~\ref{fig:incinc}).  Finally, by
  Theorem~\ref{thm:aug-solutions-minimal} and Lemma~\ref{thm:asptomm},
  we conclude that answer sets of~$\prog$ correspond to
  $\{M \in \sol(\aug(\dpa_\INCSAT({\cal T'}))) \mid sol\in M, \nexists
  M' \in \sol(\dpa_\INCSAT({\cal T'})): M' \subsetneq M\}$.

  The complexity proof sketched in \cite{BliemEtAl16b} only cares
  about the runtime being polynomial.  In fact, the algorithm can be
  carried out in linear time, following complexity proofs presented by
  Fichte~\etal~\cite{FichteEtAl17}, which leads to a worst-case
  runtime of~$2^{2^{(7k+4)}} \cdot \CCard{\prog}$.  \hfill
\end{proof}

\trash{
\begin{corollary}
Problem \AspEnum restriced to disjunctive programs can be solved in time~$f(k) \cdot \CCard{\prog}$
computing $\aug(trans(\PRIMSAT))$, where~$k$ 
refers to the treewidth of $P(\prog)$ for some computable function~$f$.
\end{corollary}}

\noindent We can now even provide a ``constructive definition'' of the
augmenting function~$\aug(\cdot)$.

\begin{proposition}
  The resulting table~$\aug(\ATab{{\AlgW}})$ obtained
  via~$\dpa_{{\AlgW}}(\cal T)$ for any TD~${\cal T}$ is equivalent to
  the table~$\cw_{{\AlgW},{\AlgW}}(\cal T)$ as given in
  Algorithm~\ref{fig:dpontd3}.
\end{proposition}
\begin{proof}[Proof (Idea)]
Intuitively, (1.) of Definition~\ref{def:augc}
concerns completeness, i.e., ensures that no row is left out during the augmentation,
and is ensured by row~7 of Algorithm~\ref{fig:dpontd3} since each~$\vec \tabval\in\ATab{{\AlgW}}$
is preserved. The second condition (2.) enforces that there is no missing
counter-witness for any witness, and the idea is that 
whenever two witnesses~$\vec \tabval, \vec v \in \tab{}$ are in a subset relation ($\alpha(\vec v) \subseteq \alpha(\vec \tabval)$)
and their corresponding linked counter-witnesses ($f_{cw}$) of the corresponding origins ($\orig$)
are in a strict subset relation, then there is some counter-witness~$c$ for~$\tabval$
if and only if $\vec c\in\tab{}$ is the successor ($\evol$) of these corresponding linked counter-witnesses.
Intuitively, we can not miss any counter-witnesses in $\cw_{{\AlgW},{\AlgW}}(\cal T)$ required by (2.),
since this required that there are two rows $\vec \tabval', \vec v'\in\tab{}'$ with $\alpha(\vec v) = \alpha(\vec \tabval)$ 
for one table~$\tab{}'$. Now let the corresponding succeeding rows~$\vec \tabval, \vec v\in\tab{}$ 
(i.e., $\vec \tabval\in\evol_t(\{\{\vec \tabval'\}\}), \vec v\in\evol_t(\{\{\vec v'\}\})$, respectively) with $\alpha(\vec v) \subsetneq \alpha(\vec \tabval)$,
$\beta(\vec v) \not\subseteq \beta(\vec \tabval)$ and $\beta(\vec v) \not\supseteq \beta(\vec \tabval)$, mark the first
encounter of a missing counter-witness. Since $\beta(\vec v)$ is incomparable to $\beta(\vec \tabval)$,
we conclude that the first encounter has to be in a table preceeding~$\tab{}$. %
To conclude, one can show that $\cw_{{\AlgW},{\AlgW}}(\cal T)$ 
does not contain ``too many'' rows, which do not fall under conditions (1.) and (2.).
\end{proof}

\noindent Theorem~\ref{thm:disjaug} works not only
for disjunctive ASP via reduction to $\MSatEnum$, where
witnesses and counter-witnesses are derived with the same table algorithm~$\INCSAT$.
In fact, one can also link counter-witnesses to witnesses by means of $\cw_{{\AlgW}, {\AlgC}}(\cdot)$,
thereby using table algorithms ${\AlgW}, {\AlgC}$ for computing witnesses
and counter-witnesses, respectively. In order to show correctness of algorithm~$\cw_{\INCSAT, \INCCOUNTERSAT}(\cdot)$ (Theorem~\ref{thm:mdp}) working for any ASP program, it is required to extend the definition of the augmenting function~$aug(\cdot)$ such that it is capable of using two different tables.

\longversion{
\begin{proposition}
Problem \AspEnum can be solved in time~$f(k) \cdot \CCard{\prog}$
computing $\cw_{\INCSAT, \INCCOUNTERSAT}(\cdot)$, where~$k$ refers to the treewidth of 
$S(\prog)$ and~$f$ is a computable function.
\end{proposition}}

\section{Additional Information on the Benchmarks}
\label{appendix:benchmarks}

\subsection{Benchmark Sets}
In this paper, we mainly presented\footnote{Benchmarks, encodings, and results are available at
\url{https://github.com/daajoe/dynasp_experiments/tree/ipec2017}.}  the Steiner tree problem
using public transport networks. 
We also considered benchmarks for counting answer sets as carried out
in earlier work~\cite{FichteEtAl17}.  {\dynaspplus{\cdot}} performs
slightly better than \dynasp{} on those instances.  Hence, we do not
report them here.

\subsubsection{Transit Graphs}
The instance graphs have been extracted from publicly available mass
transit data feeds and splited by transportation
type,~e.g., train, metro, tram, combinations. We heuristically
computed tree decompositions~\cite{AbseherMusliuWoltran17a} and
obtained relatively fast decompositions of small width unless detailed
bus networks were present.  Among the graphs considered were public
transit networks of the cities London, Bangladesh, Timisoara, and
Paris.

\subsubsection{Steiner Tree Problem}
Picture yourself as the head of a famous internet service provider,
which is about to provide high-speed internet for their
most-prestigious customers in public administration in order to
increase productivity levels beyond usual standards.  However, these
well-paying customers have to be connected via expensive fibre cables.
The good news is that the city council already confirmed that you are
allowed to use existing cable ducts, which basically adhere to the
city's transit network. 
We assumed for simplicity, that edges have unit costs, and randomly
generated a set of terminal stations --- which are compliant with the
customers --- among our transit stations (vertices).  The goal is to
search for a set of transit connections of minimal cardinality such
that the customers can be connected for example when putting fibre
cables along the transit network.

An encoding for this problem is depicted in Listing~\ref{lst:st}
and assumes a specification of the graph (via \emph{edge}) 
and the facilities (\emph{terminalVertex}) 
of our customers in public administration
as well as the number (\emph{numVertices}) of vertices.
The encoding is based on the saturation technique~\cite{EiterGottlob95}
and in fact outperformed a different encoding presented in Listing~\ref{lst:st2}
 on all our instances using both solvers, Clasp and \dynaspplus$.
At first sight, this observation seems quite surprising,
however, we benchmarked on more than 60 graphs with 10
varying decompositions for each solver variant and additional
configurations and different encodings for Clasp.}

\shortversion{\longversion}{ 
\begin{lstlisting}[caption=Encoding for \pname{St},numbers=none,label=lst:st]
%
%
%
vertex(X) :- edge(X,_).
vertex(Y) :- edge(_,Y).
edge(X,Y) :- edge(Y,X).

%
0 { selectedEdge(X,Y) } 1 :- edge(X,Y), X < Y.

%
s1(X) @$\por$@ s2(X) :- vertex(X).

%
saturate :- selectedEdge(X,Y), s1(X), s2(Y), X < Y.
saturate :- selectedEdge(X,Y), s2(X), s1(Y), X < Y.
%
%
saturate :- N #count{ X : s1(X), terminalVertex(X) }, numVertices(N).
saturate :- N #count{ X : s2(X), terminalVertex(X) }, numVertices(N).

s1(X) :- saturate, vertex(X).
s2(X) :- saturate, vertex(X).
:- not saturate.

%
#minimize{ 1,X,Y : selectedEdge(X,Y) }.
\end{lstlisting}}%

\shortversion{\longversion}{\begin{lstlisting}[caption=Alternative encoding for \pname{St},numbers=none,label=lst:st2]
%
%
edge(X,Y) :- edge(Y,X).

{ selectedEdge(X,Y) : edge(X,Y), X < Y }.

%
reached(Y) :- Y = #min{ X : terminalVertex(X) }.
reached(Y) :- reached(X), selectedEdge(X,Y).
reached(Y) :- reached(X), selectedEdge(Y,X).

:- terminalVertex(X), not reached(X).

%
#minimize{ 1,X,Y : selectedEdge(X,Y) }.
\end{lstlisting}}%

\longversion{ 
\begin{example}[Steiner Tree]%
  \label{ex:steiner}
  Let $G = (V,E)$ be a graph and $V_T \subseteq V$. A \emph{uniform
    Steiner tree} on~$G$ is a subgraph~$S_G=(V_S, E_S)$ of~$G$ such
  that $V_T \subseteq V_S$ and for each distinct pair~$v$, $w \in V_T$
  there is a path from $v$ to $w$.
  The Steiner tree problem \pname{Enum\hy{}St} asks to output all uniform
  Steiner trees. We encode \pname{Enum\hy{}St} into an \ASP program as
  follows:
  Among the atoms of our program will be an atom~$a_v$ for each
  vertex~$v\in V_T$, and an atom~$e_{vw}$ for each edge~$vw \in E$ assuming $v<w$ for an arbitrary, but fixed total ordering~$<$ among~$V$.
  Let $s$ be an arbitrary vertex~$s \in V_T$.
  We generate program~$P(G, V_T) \eqdef %
  \SB \{e_{vw}\} \hsep \rsep \optimize e_{vw} \SM vw \in E \SE \cup %
  \SB a_{v}\hsep a_w, e_{vw} \rsep a_{w}\hsep a_v, e_{vw} \SM vw \in E, v < w \SE
  \,\cup\, %
  \SB \hsep \neg a_{v} \SM v \in V_T \SE \, \cup %
  \SB a_s \hsep \SE$.  It is easy to see that the answer sets of the
  program and the uniform Steiner trees are in a one-to-one
  correspondence.

\end{example}%

\subsubsection{Other Graph Problems}
We refer to the technical report~\cite{FichteHecherMoraketal2016report_} for comprehensive
benchmark results using common graph problems and comparing
different table algorithms of \dynasp{} and \dynaspplus{}.
There we investigated various graph problems including, but not limited to, variants of 
graph coloring, dominating set and vertex cover. Note that, however, the cluster setup was slightly different.

\subsection{Benchmark Environment}
The experiments presented ran on an Ubuntu~16.04.1 LTS Linux cluster of 3 nodes
with two Intel Xeon E5-2650 CPUs of 12 physical cores each at 2.2 Ghz
clock speed and 256 GB RAM.  All solvers have been compiled with gcc
version 4.9.3 and executed in single core mode.
}

\label{lastpage}
\end{document}